\documentclass[aps]{revtex4}

\usepackage[english]{babel}
\usepackage[utf8x]{inputenc}
\usepackage{amsmath}
\usepackage{amssymb}
\usepackage{amsthm}
\usepackage{graphicx}
\usepackage{verbatim}
\usepackage{hyperref}
\usepackage[colorinlistoftodos]{todonotes}
\usepackage{color}

\usepackage{epstopdf} 

\input{Qcircuit}

\def\ket#1{| #1 \rangle}
\def\bra#1{\langle #1 |}

\def\bk#1#2{\langle #1 |#2\rangle}
\def\cal#1{\mathcal{#1}}

\makeatletter
\newcommand{\notprop}{\propto\kern-1\@ptsize pt \diagup}
\makeatother

\newtheorem{remark}{Remark}
\newtheorem{corollary}{Corollary}
\newtheorem{theorem}{Theorem}
\newtheorem{prop}{Proposition}
\newtheorem{claim}{Claim}
\newtheorem{lemma}{Lemma}

\newcommand{\etal}{\textit{et~al.}}

\newcommand{\ie}{\textit{i.e.}}

\begin{document}

\title{Classification of transversal gates in qubit stabilizer codes}

\author{Jonas T. Anderson}
\affiliation{D\'epartement de Physique, Universit\'e de Sherbrooke, Sherbrooke, Qu\'ebec, J1K 2R1, Canada}

\author{Tomas Jochym-O'Connor}
\affiliation{Institute for Quantum Computing, Department of Physics \& Astronomy, University of Waterloo, Waterloo, Ontario, N2L 3G1, Canada}

\begin{abstract}
This work classifies the set of diagonal gates that can implement a single or two-qubit transversal logical gate for qubit stabilizer codes. We show that individual physical gates on the underlying qubits that compose the code are restricted to have entries of the form~$e^{i \pi c/2^k}$ along their diagonal, resulting in a similarly restricted class of logical gates that can be implemented in this manner. Moreover, we show that all diagonal logical gates that can be implemented transversally by individual physical diagonal gates must belong to the Clifford hierarchy. Furthermore, we can use this result to prove a conjecture about transversal gates made by Zeng \etal~in 2007.
\end{abstract}

\maketitle

\section{Introduction}
Any physical realization of a quantum computing device will be subject to physical noise processes leading to potential computational errors. As such, quantum error correction is used to protect the information using multiple physical systems to encode a logical quantum state~\cite{Shor:1995a, Steane:1996b, Bennett:1996a, Knill:1997a}. Quantum error correction will play a central role in any fault-tolerant implementation of a quantum computer, yet it is of paramount importance that the fundamental quantum operations such as state preparation, error syndrome extraction and correction, state measurement, and state manipulation are done in a manner that does not propagate errors throughout the system~\cite{Shor:1996a, Preskill:1998c, Knill:1998a, Aharonov:1997a, Knill:2004a}. In this work we focus on state manipulation, or quantum gate application. Transversal gates, that is, logical gates that are a result of the application of individual local quantum gates on qubits forming the quantum error correcting code, provide the most natural form of fault-tolerant quantum logic. Therefore, developing quantum error correcting codes that have transversal gate sets are of prime importance for quantum fault-tolerance. However, as first shown by Zeng~\etal~for stabilizer codes~\cite{Zeng:2007a}, and then further generalized by Eastin and Knill~\cite{Eastin:2008a} for any quantum error correcting code, there exists no quantum error correcting code that has a set of universal transversal gates. Additionally, Bravyi and Kronig~\cite{Bravyi2013} showed that for a $D$-dimensional local stabilizer code with large distance that only gates from the Clifford Hierarchy at level $D-1$ (or lower) can be applied transversally. Their result also applies to more general local unitaries, not just transversal gates. 

With these constraints in mind, there has been a push in the research community towards methods to side-step these gate restrictions for a single quantum error correcting code. Techniques that allow for the fault-tolerant application of a set of universal quantum gates involve quantum code manipulation through gate fixing~\cite{Paetznick:2013a, Bombin:2013a}, partial transversality~\cite{Jochym:2014a}, or code conversion~\cite{Anderson:2014a}. While these results show promise, practical techniques for implementing fault-tolerant universal gate logic without having to use techniques such as magic state distillation, with its high qubit overhead, would be useful for further improvements. 

Recent results in the area of quantum gate decomposition have focused on expressing an arbitrary single-qubit quantum gate as a sequence of Hadamard ($H$) and $V$ gates, where $V = \text{diag}(1+2i, 1-2i)/\sqrt{5}$~\cite{Bocharov:2013a}. Therefore, the discovery of quantum error correcting codes that allow for the application of~$V$ in a transversal manner could potentially led to adaptation of the above mentioned techniques for universal fault-tolerant gate application without state distillation for these proposed gate decompositions.

Recently, a parallel work by Pastawski and Yoshida~\cite{Pastawski:2014a} showed many exciting results pertaining to fault-tolerant operations in topological stabilizer codes. They also proved that families of stabilizer codes with a finite loss threshold must have transversal gates in the Clifford hierarchy. Furthermore, they show that higher loss thresholds impose greater restrictions on the level in the Clifford hierarchy at which transversal gates can be implemented. Our result does not need a finite loss threshold or a family of codes to be applicable. Additionally, our result applies to transversal gates between two like codes; however, this result only applies to qubits and restricts the transversal gates to being in the Clifford hierarchy (our result does not specify the level).

The main result of our paper is that for quantum qubit stabilizer codes, the only diagonal gates that can be implemented transversally are those whose entries along the diagonal are of the form~$e^{i\pi c/2^k}$, for some power of~$k$ depending on the choice of code. This result holds both for single and two-qubit gates, and moreover we show that all such gates must be contained within the Clifford hierarchy. Moreover, as Zeng~\etal~ showed~\cite{Zeng:2007a}, any transversal non-trivial single-qubit logical gate for a qubit stabilizer code must result from the application of diagonal gates along with local Clifford operations and potential swapping of qubits. Therefore, our result classifies all transversal single-qubit logical gate operations up to local Clifford equivalences and relabelling of qubits. Additionally, our result classifies all transversal diagonal single-qubit logical gates that can map one stabilizer code to another stabilizer code. It is worth noting that the Reed-Muller family of quantum codes provides a means of implement any of these diagonal transversal gates, where changing to higher order in the code family allows for the implementation of diagonal logical gates with finer angles, all of which are in the Clifford hierarchy and of the form~$e^{i\pi c/2^k}$.

\section{Stabilizer codes and transversal logical gates}

\subsection{The stabilizer formalism}

We begin by reviewing the stabilizer formalism~\cite{Gottesman:1996a,Gottesman:1999b}. The Pauli matrices are defined as follows:
\begin{align*}
I = 
\begin{pmatrix}
1 & 0 \\
0 & 1 
\end{pmatrix}, \qquad
X = 
\begin{pmatrix}
0 & 1 \\
1 & 0 
\end{pmatrix}, \qquad
Y = 
\begin{pmatrix}
0 & -i \\
i & 0 
\end{pmatrix}, \qquad
Z = 
\begin{pmatrix}
1 & 0 \\
0 & -1 
\end{pmatrix}.
\end{align*}

The Pauli group on $n$ qubits~$\cal{P}_n$ is generated by the above Pauli matrices on each of the $n$~qubits. Given a set of independent commuting elements~$\{ P_1, \hdots, P_{n-k} \}$ from the Pauli group~$\cal{P}_n$, the group generated by these elements modulo overall phase factors~$\{1, i, -1, -i \}$, denoted~$\cal{S} = \langle G_1, \hdots , G_{n-k} \rangle$ is the~\textit{stabilizer} of a quantum code on $n$~qubits:~$Q = \{ \ket{\psi} \ | \ g\ket{\psi}  = \ket{\psi} \ \forall \ g \in \cal{S} \}$. The quantum code~$Q$ corresponds to the intersection of the~``+1"~eigenspaces of all of the~$(n-k)$ generators and has dimension size~$2^k$, that is, it will encode~$k$ logical qubits. Logical operators are the elements of the \textit{normalizer}~of~$\cal{S}$, $\cal{N}(\cal{S}) = \{ U \in U(2^n) \ | \ U\cal{S}U^\dagger = \cal{S} \}$, that are not trivially in the stabilizer~$\cal{S}$, that is, $\cal{N}(\cal{S}) / S$. The \textit{distance} of the code~$Q$ is defined as~$d = \text{min} \{ \text{wt}(P) \ | \ P \in \cal{P}_n, \ P \in \cal{N}(\cal{S}) / S \}$, where the \textit{weight}~$\text{wt}(P)$ is defined as the number of non-identity elements in the Pauli operator~$P$. An error-detecting quantum stabilizer code~$Q$ is any stabilizer code whose distance~$d\ge 2$. Throughout the remainder of this work, a stabilizer code will refer to an error-detecting quantum stabilizer code unless otherwise specified.

\subsection{Outline of proof}

In 2007, Zeng \etal~\cite{Zeng:2007a} showed that unitary, single-qubit logical transversal operators in qubit stabilizer codes were of the form:

\begin{equation}
U = L\left(\bigotimes_{j=1}^{n}\mbox{diag}(1,e^{i\pi\theta_j})\right)R^{\dagger}P_{\pi}.
\end{equation}
Here $L,R^{\dagger}$ are tensor products of local Clifford operations and $P_{\pi}$ is a coordinate permutation (a set of SWAP gates). Notice that if an $[[n,k,d]]$ stabilizer code exists which implements $U$ transversally, then up to local Clifford equivalences, an $[[n,k,d]]$ stabilizer code exists which implements $U=\bigotimes_{j=1}^{n}\mbox{diag}(1,e^{i\pi \theta_j})$ transversally. In this work we look at the restrictions on these diagonal, transversal gates. 

First, we prove that all diagonal gates are of the form $\mbox{diag}(1,e^{i\pi c/2^k})$ when $\theta_1=\theta_j\ \forall \ j$ for some natural number $k$. We prove this first for CSS codes and then for general stabilizer codes.

Then, we prove the case when $\theta_i$ not necessarily equal to $\theta_j$. We first show (proved in the appendix) that irrational angles must cancel each other out and therefore add nothing. We can then restrict to rational angles $\theta$. We prove a {\it decompression lemma} which allows us to reduce this case to the uniform-$\theta$ case, and our proof carries through as before. 

At this point, we have shown that all transversal, unitary gates on one codeblock are of the form  
\begin{equation}
U = L\left(\bigotimes_{j=1}^{n}\mbox{diag}(1,e^{i\pi\theta_j})\right)R^{\dagger},
\end{equation}
where $\theta_j = c_j/2^{k_j}$ and $c_j,\ k_j$ are integers. This is up to additional operators which cancel out and apply the logical identity operator and proves a conjecture made by Zeng \etal~that all transversal gates in qubit stabilizer codes are in the Clifford hierarchy. 

\section{Strongly Transversal $Z$ rotations}

Inspired by our previous discussion, we will focus on implementing rotations $Z(\theta)$, that is rotations about the $Z$-axis by some angle $\theta$. For qubits this rotation is given by a diagonal matrix
 
\begin{equation}
	A = 
    \begin{bmatrix}
    e^{i\pi\theta_1} & 0\\
    0 			   & e^{i\pi\theta_2}
    \end{bmatrix}
    =
    e^{i\pi\theta_1}\begin{bmatrix}
    1 & 0\\
    0 			   & e^{i\pi(\theta_2-\theta_1)}
    \end{bmatrix}.
\end{equation}
Up to a global phase, we need only consider rotations of the form
\begin{equation}
	A = 
	\begin{bmatrix}
    1 & 0\\
    0 			   & e^{i\pi\theta}
    \end{bmatrix}
    \equiv Z(\theta),
\end{equation}
where we are using the above equation as the definition of a single-qubit $Z(\theta)$~rotation of angle~$\pi \theta$ (we shall assume for the remainder of this work that the angular rotations are rational multiples of~$\pi$, as discussed in detail below).

In this work, we study constraints on transversal implementations of logical $Z(\theta)$ rotations. A transversal $Z$ rotation is defined as
\begin{equation}
Z_T(\theta) := Z(\theta_1)\otimes Z(\theta_2)\otimes...\otimes Z(\theta_n)
\end{equation}

Before considering the most general form of transversal gate outlined above, we first focus on the case when all physical qubits undergo the same rotation~$\theta$, that is, we require that the logical implementation be strongly transversal ($Z_L(\theta') = Z(\theta)^{\otimes n}$) with $n$ being the number of physical qubits. While each single-qubit rotation is a rotation by the same angle, we do not require that the logical $Z$ applies the same rotation to the logical qubit. 

\subsection{CSS codes}

A CSS code~\cite{Steane:1996a,Calderbank:1996a} is a stabilizer code~\cite{Gottesman:1996a,Gottesman:1999b} whose generators can be separated into two sets, the $X$~stabilizers composed of only Pauli~$X$ operators and the $Z$~stabilizers, that is~$\cal{S} = \langle G_{X_1}, \hdots, G_{|G_X|}, G_{Z_1}, \hdots, G_{|G_Z|} \rangle$, where $|G_X|$ and $|G_Z|$ refer to the number of~$X$ and $Z$~stabilizers, respectively.

\begin{theorem}
A nontrivial CSS code can have only strongly transversal $Z(\theta)$ rotations which are of the form $Z( a/2^k)$.  
\end{theorem}

It is worth noting that Reed-Muller codes exist which have any $Z(1/2^k)$ gate transversally. Additionally, these gates are all in the Clifford hierarchy~\cite{Gottesman:1999d}. 

\begin{proof}
We can express the logical states for CSS codes as follows:

\begin{align}\label{eq:logcss0}
|0_L\rangle &= \frac{1}{2^{|G_X|/2}}\prod_{i}(I + G_{X_i})|0\rangle^{\otimes n},\\
\label{eq:logcss1}
|1_L\rangle &= X_L|0_L\rangle,
\end{align}
where $i$ runs over all $X$ stabilizer generators $G_{X_i}$. These are codestates, as~$\prod_{i}(I + G_{X_i})$ projects onto the codespace and the state~$\ket{0_L}$ must be an eigenstate of $Z_L$ since the logical operator must consist only of~$Z$ operators due to it being a CSS code.
 
To determine if a CSS code has a logical $Z$ rotation we only need to look at the $X$ stabilizers ($\mathcal{S}_X$) and $X$ logical operators ($\mathcal{L}_X$). (We are assuming that $X(Z)$ logical consists only of $X(Z)$ Pauli operators.) 

The constraints come from the following properties of logical $Z(\theta)$:  
\begin{equation}
Z_L(\theta')|0_L\rangle = \frac{1}{2^{|G_X|/2}}Z(\theta)^{\otimes n}\prod_{i}(I + G_{X_i})|0\rangle^{\otimes n} = |0_L\rangle,
\end{equation}

and

\begin{equation}
Z_L(\theta')|1_L\rangle = Z_L(\theta')X_L|0_L\rangle = Z(\theta)^{\otimes n} X_L |0_L\rangle = e^{i\pi\theta}X_L Z(\theta')^{\otimes n}|0_L\rangle = e^{i\pi\theta'}|1_L\rangle.
\end{equation}

Here we assume that the logical~$Z_L$ yields no global phase on the logical~$\ket{0_L}$ state. In general, a valid logical~$Z_L$ operation can be diagonal in the logical basis; however this additional freedom does not provide additional freedom in the choice of individual rotations~$Z(\theta)$, as shall be discussed at the conclusion of the proof.

We will find it more convenient to rewrite these constraints as
\begin{align*}
&Z(\theta)^{\otimes n}\prod_{i}(I + G_{X_i})|0\rangle^{\otimes n}\\
&= |g_0\rangle+Z(\theta)^{\otimes n}\left(\sum_{i_1}|g_{i_1}\rangle + \sum_{i_1<i_2}|g_{i_1}\oplus g_{i_2}\rangle + ... + \sum_{i_1<i_2<...<i_{|G_X|}}|g_{i_1}\oplus...\oplus g_{i_{|G_X|}}\rangle\right) \\
&= |g_0\rangle+\sum_{i_1}e^{i\theta|g_{i_1}|}|g_{i_1}\rangle + \sum_{i_1<i_2}e^{i\theta|g_{i_1} \oplus g_{i_2}|}|g_{i_1}\oplus g_{i_2}\rangle + ... +\\
& \sum_{i_1<i_2<...<i_{|G_X|}}e^{i\theta|g_{i_1}\oplus...\oplus g_{i_{|G_X|}}|}|g_{i_1}\oplus...\oplus g_{i_{|G_X|}}\rangle \\
&= |g_0\rangle+\sum_{i_1}|g_{i_1}\rangle + ... + \sum_{i_1<i_2<...<i_{|G_X|}}|g_{i_1}\oplus...\oplus g_{i_{|G_X|}}\rangle
\end{align*}

and

\begin{align*}
&Z(\theta)^{\otimes n}|1\rangle =\\
&e^{i\pi\theta |g_L|}\left(|g_L\rangle + \sum_{i}|g_L\oplus g_i\rangle + ... + \sum_{i_1<i_2<...<i_{|G_X|}}|g_L\oplus g_{i_1}\oplus...\oplus g_{i_{|G_X|}}\rangle\right).
\end{align*}

Here $g_0$ is the all-zeros string, $g_i(g_L)$ is a binary string corresponding to $G_{X_i}(X_L)$, $|g_L|$ is the hamming weight of $g_L$, and $\oplus$ corresponds to the bitwise XOR. For $Z(\theta)$ to be nontrivial we require that $\theta|g_L| \ne 0\bmod 2$. Rows of $G_{X_i}$ and $X_L$ can be expressed as binary strings with the association $X\rightarrow 1, I\rightarrow 0$.

Since each term in the above equations is a different binary string, the constraints must be satisfied independently.

The constraints on logical 0 give us

\begin{align*}
 \theta|g_{i_1}| &= 0 \bmod 2\\
 \theta|g_{i_1}\oplus g_{i_2}| &= 0 \bmod 2\\
 \vdots& \\
 \theta|g_{i_1}\oplus...\oplus g_{i_{|G_X|}}| &= 0 \bmod 2,\\
\end{align*}

while the constraints on logical 1 give us
\begin{align*}
 \theta|g_L| &= a \bmod 2\\
 \theta|g_L\oplus g_{i_1}| &= a \bmod 2\\
 \vdots& \\
\theta|g_L\oplus g_{i_1}\oplus...\oplus g_{i_{|G_X|}}| &= a \bmod 2\\
& \forall \ 0<i_1<i_2<...<i_{|G_X|}\le|G_X|.
\end{align*}

We begin by making some observations on the above equations to rule out certain values of~$\theta$. 

\begin{enumerate}
  \item First, notice that if $\theta$ is irrational these equations can never be satisfied since $n\theta=p \equiv 0\bmod 2\implies \theta=\frac{2t}{n}\in \mathbb{Q}$. We can therefore restrict our attention to rational angles ($\theta = \frac{p}{q} \in \mathbb{Q}$). Without loss of generality, we can assume this fraction is irreducible and in the range $(0,2]$.
  \item Notice that the value of $p\in \mathbb{Z}$ is not important; only whether it is even or odd.
  
  If $p$ is even we have $|\cdot| = 0\bmod q$, if $p$ is odd and $q$ is even we have $|\cdot| = 0\bmod q$, and if $p$ is odd and $q$ is odd we have $|\cdot| = 0\bmod 2q$. If $p$ and $q$ were both even this would violate our assumption that the fraction is irreducible. The case where both are odd is more restrictive and since we are ultimately trying to find the most general $\theta$ allowable, so we will assume $p$ is even. A proof of the other cases follows in the same manner. 
  \item We can express these constraints as conditions on overlap similarly to Bravyi and Haah~\cite{Bravyi:2012a} by noting that
  \begin{equation}
  	|g_{1}\oplus...\oplus g_{n}| = \sum_{i=1}^{n}|g_i| -2\sum_{i<j}|g_{i}\wedge g_{j}|+...+(-2)^{n-1}\sum_{i<...<n}|g_i\wedge...\wedge g_n|.
    \label{eq:BravyiHaah}
  \end{equation}
  Here $\wedge$ is the bitwise AND. 
\end{enumerate}

With these observations and the assumption that $p$ is even, we can express the constraints as
\begin{align*}
 |g_{i_1}| &= 0 \bmod q\\
 |g_{i_1}| + |g_{i_2}| -2|g_{i_1}\wedge g_{i_2}| &= 0 \bmod q\\
 \vdots& \\
 \sum_{i_1}^{n}|g_{i_1}| -2\sum_{i_1<i_2}|g_{i_1}\wedge g_{i_2}|+...+(-2)^{n-1}\sum_{i_1<...<i_{|G_X}}|g_{i_1}\wedge...\wedge g_{i_{|G_X|}}| &= 0 \bmod q\\
 \\
 |g_L| &= b \bmod q\\
 |g_L| + |g_{i_1}| -2|g_L\wedge g_{i_1}| &= b \bmod q\\
\vdots& \\
 \forall \ 0<i_1<i_2<...<i_{|G_X|}\le|G_X|&
\end{align*}

We can see that these equations are not independent since the requirement that $|g_{i}| = 0\bmod q$, implies 
\begin{equation}
|g_{i_1}| + |g_{i_2}| -2|g_{i_1}\wedge g_{i_2}| = 0 \bmod q \implies 2|g_{i_1}\wedge g_{i_2}| = 0 \bmod q. 
\end{equation}

Using this, we can express the above constraints as overlap conditions
\begin{align*}
 |g_{i}| &= 0 \bmod q, \forall 0<i\le|G_X|\\
 2|g_{i_1}\wedge g_{i_2}| &= 0 \bmod q\\
 4|g_{i_1}\wedge g_{i_2}\wedge g_{i_3}| &= 0 \bmod q\\
 \vdots& \\
 (2)^{|G_X|-1}|g_{i_1}\wedge...\wedge g_{i_{|G_X|}}| &= 0 \bmod q\\
  |g_{L}| &\ne 0 \bmod q\\
   2|g_{i_1}\wedge g_{L}| &= 0 \bmod q\\
 4|g_{i_1}\wedge g_{i_2}\wedge g_L| &= 0 \bmod q\\
 \vdots& \\
 (2)^{|G_X|}|g_{i_1}\wedge...\wedge g_{i_{|G_X|}} \wedge g_L| &= 0 \bmod q\\
 \forall 0<i_1<i_2<...<i_{|G_X|}\le|G_X|&,
\end{align*}
with $i_1,...,i_{|G_X|}$ now a sum over stabilizer generators ($g_i$) and $g_X$. We have also dropped the minus sign since it has no effect. For the logical operator to be nontrivial, we have assumed that $a,b\ne 0$. Notice that the $0 \bmod q$ conditions are independent constraints. 

Observe that if $q$ has only even prime factors (\ie~$q=2^t$ for some integer $t$) then all higher-order overlap conditions will, at some point, become trivial. For example, if $q=2^{k}$ overlap conditions will be trivial for any $k+1$ or more rows and Reed-Muller codes exist which have any $Z(1/2^k)$ gate transversally. In fact, since the transversal gates form a group, Reed-Muller codes exist which have any $Z(c/2^k)$ (where $c$ is an integer) gate transversally. Therefore the existence of transversal gates is already solved in the positive for that case. 

In what follows we will assume that $q$ has a least one odd prime factor and that $|g_L|\ne 0 \bmod q_o$ for at least one such $q_o$ (we will choose this $q_o$). As mentioned above, the case where $q$ has only even prime factors ($q=2^t$) is already solved. If $|g_L|=0\bmod q_o$ for all odd prime factors, then $Z(\theta)=e^{i\pi Z|g_L|/q}=e^{i\pi Z a/2^{k}}$ for some positive integer $k$. Here $a\equiv |g_L| \bmod 2^k$. In this case, the odd prime factors add nothing, and we could apply the same logical operator by using $Z(a/2^{k})$ instead of $Z(a/q)$. Since this case is already solved for, we assume $|g_L|\ne 0 \bmod q_o$ for at least one such $q_o$. Observe that if $q$ has at least one odd prime factor $q_o$, then all overlap conditions are nontrivial. We can write $q=q_o\cdot q_{P/o}$ where $q_{P/o}$ is the product of the other prime factors of $q$. Since $|g| = 0,1 \bmod q \implies |g| = 0,1 \bmod q_o$, we can write a weaker set of overlap conditions as

\begin{align*}
 |g_{i}| &= 0 \bmod q_o, \forall 0<i\le|G_X|\\
 |g_{i_1}\wedge g_{i_2}| &= 0 \bmod q_o\\
 |g_{i_1}\wedge g_{i_2}\wedge g_{i_3}| &= 0 \bmod q_o\\
 \vdots& \\
 |g_{i_1}\wedge...\wedge g_{i_{|G_X|}}| &= 0 \bmod q_o\\
  |g_{L}| &\ne 0 \bmod q_o\\
   |g_{i_1}\wedge g_{L}| &= 0 \bmod q_o\\
 |g_{i_1}\wedge g_{i_2}\wedge g_L| &= 0 \bmod q_o\\
 \vdots& \\
 |g_{i_1}\wedge...\wedge g_{i_{|G_X|}} \wedge g_L| &= 0 \bmod q_o\\
 \forall \ 0<i_1<i_2<...<i_{|G_X|}\le|G_X|&
\end{align*}

\begin{remark}
\normalfont
We made the assumption that the logical~$X$ operator was composed of a set of individual~$X$ operators on a collection of qubits characterized by the bit string~$g_L$, where $g_L(i) = 1$ if $X_L$ performs the operation~$X$ at qubit~$i$. However in theory, $X_L$ could also be comprised of $Z$~(or~$Y$) operations as well. A particular~$Z$~(or~$Y$) gate could introduce a phase on some of the state vectors in the expansion of the logical~$\ket{1_L}$, yet these phases must be preserved by the action of~$Z(\theta)^{\otimes n }$. Since these diagonal rotations will not change the form of the computational basis state, they will only introduce a phase. In that manner, the presence of $Z$~(or~$Y$) operations in the logical~$X_L$ gate will not change the set of algebraic conditions for the physical rotations~$Z(\theta)$.
\end{remark}

\begin{remark}
\normalfont
We made the assumption that the individual rotation on the physical qubits, $Z(\theta)$, were of the form~$\text{diag}(1, e^{i\theta})$, however in full generality the diagonal gates can be of the form~$\text{diag}(e^{i \varphi}, e^{i\theta})$. The resulting conditions on the transformation of the logical states~$\ket{0_L}$ and $\ket{1_L}$ will have the form:
\begin{align*}
Z(\theta)^{\otimes n}\ket{0_L} &= Z(\theta)^{\otimes n}\left( \ket{g_0} + \sum_{i_1}|g_{i_1}\rangle + \sum_{i_1<i_2}|g_{i_1}\oplus g_{i_2}\rangle + ... + \sum_{i_1<i_2<...<i_{|G_X|}}|g_{i_1}\oplus...\oplus g_{i_{|G_X|}}\rangle\right) \\
&= e^{i \varphi n} |g_0\rangle+\sum_{i_1}e^{i\theta|g_{i_1}| + i \varphi (n-|g_i|)}|g_{i_1}\rangle + \sum_{i_1<i_2}e^{i\theta|g_{i_1} \oplus g_{i_2}| + i\varphi (n-|g_{i_1} \oplus g_{i_2}|)}|g_{i_1}\oplus g_{i_2}\rangle + ... +\\
& \sum_{i_1<i_2<...<i_{|G_X|}}e^{i\theta|g_{i_1}\oplus...\oplus g_{i_{|G_X|}}|+ i \varphi(n-|g_{i_1}\oplus...\oplus g_{i_{|G_X|}}|)}|g_{i_1}\oplus...\oplus g_{i_{|G_X|}}\rangle \\
&= e^{i\varphi n }\left( \ket{g_0} + \sum_{i_1}|g_{i_1}\rangle + \sum_{i_1<i_2}|g_{i_1}\oplus g_{i_2}\rangle + ... + \sum_{i_1<i_2<...<i_{|G_X|}}|g_{i_1}\oplus...\oplus g_{i_{|G_X|}}\rangle\right)
\end{align*}
and
\begin{align*}
Z(\theta)^{\otimes n} \ket{1_L} = e^{i \varphi n + i(\theta-\varphi) |g_L| }\left( \ket{g_L} + \sum_{i_1}|g_L \oplus g_{i_1}\rangle + ... + \sum_{i_1<i_2<...<i_{|G_X|}}|g_L \oplus g_{i_1}\oplus...\oplus g_{i_{|G_X|}}\rangle\right).
\end{align*}
The constraints can then be shown to have the form:
\begin{align*}
 (\theta-\varphi)|g_{i_1}| &= 0 \bmod 2\\
 2(\theta-\varphi) |g_{i_1}\wedge g_{i_2}| &= 0 \bmod 2\\
 \vdots& \\
 2^{ |G_X| -1}(\theta-\varphi)|g_{i_1}\wedge \hdots \wedge g_{i_{|G_X|}}| &= 0 \bmod 2\\
 \\
 (\theta-\varphi)|g_L| &\ne 0 \bmod 2\\
 2(\theta-\varphi)|g_L\wedge g_{i_1}| &= 0 \bmod 2\\
 \vdots& \\
2^{|G_X|}(\theta-\varphi)|g_L\wedge g_{i_1}\wedge...\wedge g_{i_{|G_X|}}| &= 0 \bmod 2\\
& \forall 0<i_1<i_2<...<i_{|G_X|}\le|G_X|,
\end{align*}
which are the same constrains on the difference of the phases~$(\theta-\varphi)$ as the case when~$\varphi=0$. Therefore, an arbitrary global phase can be introduced on the individual rotations of the form~$\text{diag}(1, e^{i\theta})$ which are allowed in the CSS construction. 
\end{remark}

In what follows, we will attempt to find the smallest binary matrix (in terms of number of rows) which satisfies all overlap conditions.    

\subsubsection{Existence of binary matrix}\label{sssec:Existence}

We define a binary matrix $M$ with each row given by a binary string $g_i$. For $i\in\{1,...,n\}$, this is the binary string corresponding to an $X$ stabilizer generator. We will refer to this as the $X$ stabilizer submatrix, $S_X$. Notice that each row is independent. We refer to a code as nontrivial if the distance is at least 2. The $X$ stabilizer submatrix is said to be nontrivial if no columns containing only zero exist. This is a necessary and sufficient condition for $Z$ error detection. The remaining rows of $M$ are given by binary strings corresponding to $X$ logical operators. We will consider the case of a single logical operator and show that no nontrivial matrix $M$ with at least one $X$ logical operator exists, such that all rows satisfy the overlap conditions derived above~\footnote{This part of our proof uses techniques developed in \cite{Bravyi:2012a}}.

Let us now try to find the smallest number of rows in $S_X$ such that the overlap conditions are satisfied. In what follows, we start with the assumption that $|g_L|\ne 0\bmod q_o$ and derive a contradiction.

If $S_X$ is nontrivial, it must contain at least one row. If $S_X$ contains a single row ($g_1$), it must be the ``all-ones'' row and be a multiple of $q_o$, since $|g_1|=0\bmod q_o$. Also, $|g_1\wedge g_L|=0 \bmod q_o \implies |g_L|=0 \bmod q_o$ and hence a contradiction.

If $S_X$ is nontrivial and has two rows, all columns of $S_X$ are of one of three types:

\begin{equation}
	a = \begin{bmatrix}
    1\\0
    \end{bmatrix},
    b = \begin{bmatrix}
    0\\
    1 
    \end{bmatrix},
    c = \begin{bmatrix}
    1\\
    1 
    \end{bmatrix}.
\end{equation}

We will refer to the combination of all columns of type $a,b, c$, by the matrix $A,B,C$, respectively.

If we have a logical operator $g_L$, then 

\begin{align*}
 |g_1\wedge g_L| &= w_A + w_C = 0 \bmod q_o,\\
 |g_2\wedge g_L| &= w_B + w_C = 0 \bmod q_o,\\
 |g_1\wedge g_2\wedge g_L| &= w_C = 0 \bmod q_o,\\
 |g_L| &= w_A + w_B + w_C \ne 0 \bmod q_o.
\end{align*}

Here, $w_A$ is the overlap of $A$ and $g_L$. The first three constraints imply that $w_A, w_B, w_C = 0 \bmod q_o$ which imply $|g_L| = 0\bmod q_o$ and hence a contradiction. 

Now, if $S_X$ has three rows, we will have 7 independent ($0 \bmod q_o$) conditions on 7 variables ($w_A,...,w_G$),\begin{equation*}
\begin{array}{l|ccccccc}
& w_A & w_B & w_C &w_D &w_E &w_F &w_G\\
\hline
|g_1\wedge g_L| & 1 & 0 & 0 & 1 & 1 & 0 & 1\\
|g_2\wedge g_L| & 0 & 1 & 0 & 1 & 0 & 1 & 1\\
|g_3\wedge g_L| & 0 & 0 & 1 & 0 & 1 & 1 & 1\\
|g_1\wedge g_2 \wedge g_L| & 0 & 0 & 0 & 1 & 0 & 0 & 1\\
|g_1\wedge g_3 \wedge g_L| & 0 & 0 & 0 & 0 & 1 & 0 & 1\\
|g_2\wedge g_3 \wedge g_L| & 0 & 0 & 0 & 0 & 0 & 1 & 1\\
|g_1\wedge g_2 \wedge g_3 \wedge g_L|& 0 & 0 & 0 & 0 & 0 & 0 & 1\\
\end{array}
\end{equation*}

where

\begin{equation*}
	a = \begin{bmatrix}
    1\\0\\0
    \end{bmatrix},
    b = \begin{bmatrix}
    0\\1\\0 
    \end{bmatrix},
    c = \begin{bmatrix}
    0\\0\\1
    \end{bmatrix},
    d = \begin{bmatrix}
    1\\1\\0
    \end{bmatrix},
    e = \begin{bmatrix}
    1\\0\\1 
    \end{bmatrix},
    f = \begin{bmatrix}
    0\\1\\1
    \end{bmatrix},
    g = \begin{bmatrix}
    1\\1\\1
    \end{bmatrix}.
\end{equation*}

Therefore, the conditions on the overlap variables~$(w_A, \hdots, w_G)$ can be expressed as a matrix equation as follows, where the righthand vector is expressed $\bmod~q_o$:
\begin{align*}
\begin{pmatrix}
1 & 0 & 0 & 1 & 1 & 0 & 1\\
0 & 1 & 0 & 1 & 0 & 1 & 1\\
0 & 0 & 1 & 0 & 1 & 1 & 1\\
0 & 0 & 0 & 1 & 0 & 0 & 1\\
0 & 0 & 0 & 0 & 1 & 0 & 1\\
0 & 0 & 0 & 0 & 0 & 1 & 1\\
0 & 0 & 0 & 0 & 0 & 0 & 1
\end{pmatrix}
\begin{pmatrix}
w_A \\ w_B \\ w_C \\ w_D \\ w_E \\ w_F \\ w_G 
\end{pmatrix}
=
\begin{pmatrix}
0 \\ 0 \\ 0 \\ 0 \\ 0 \\ 0 \\ 0
\end{pmatrix}.
\end{align*}

This implies that~$w_i = 0 \bmod q_o  \ \forall \ i$, which as in the case of an $X$~generator matrix with two rows will contradict the assumption that~$|g_L| = \sum_i w_i \ne 0 \bmod q_o$.

Furthermore, if $S_X$ has $m$ rows we will have $2^{m}-1$ independent overlap constraints (the number of nonzero column vectors of size $m$) all requiring that a sum of weights $w_i$ must equal $0\bmod q_o$. There will be $2^{m}-1$ overlap variables $w_i$ which will each be forced to equal ($0 \bmod q_o$) to satisfy these constraints; we also have a constraint on the overall sum of these variables which must not be equal to ($0 \bmod q_o$). Therefore, no binary matrix with $k$ rows can satisfy all $k$-overlap conditions and have overlap which is not equal to $0 \bmod q_o$ with $g_L$.  

Our proof holds if additional logical operators are included, since these conditions must be satisfied by {\it each} logical operator and we showed that they cannot be satisfied by even a single logical operator. 
\end{proof}

It is worth noting that the restriction on the set of rotations that can be applied to the individual qubits of a CSS~code will impose a restriction on the set of logical rotations that can be applied. This shows a strong connection to the Clifford hierarchy. The Clifford hierarchy is defined recursively, where the first level of the hierarchy on $n$~qubits is defined as the Pauli operators on $n$~qubits, denoted~$\cal{C}_n^{(1)} = \cal{P}_n$. Higher levels~($k\ge 2$) of the Clifford hierarchy are then defined as follows:
\begin{align*}
\cal{C}_n^{k} = \{ U \in U (2^n) \ | \ U P U^{\dagger} \in \cal{C}_n^{(k-1)} \ \forall P \in \cal{P}_n \},
\end{align*}
that is, a unitary~$U$ in the $k$-th level of the Clifford hierarchy maps by conjugation the Pauli operators on $n$~qubits to an element in the $(k-1)$-th level of the Clifford hierarchy. Namely, the second level of the Clifford hierarchy is the Clifford operators, mapping Pauli operators to Pauli operators. It is worth noting that each level of the Clifford hierarchy contains all lower levels of the Clifford hierarchy, that is~$\cal{C}_n^{(p)} \subsetneq \cal{C}_n^{(q)}$, if $p<q$. 

\begin{prop}
Let $A=Z(\theta)$ be a diagonal single-qubit operator. If $\theta = c/2^k$, for any integer $k\ge 0$ where $\theta$ is in its most reduced form, then $A \in \cal{C}_1^{(k+1)}$. Otherwise, $A$~is not in the Clifford hierarchy, that is $A \notin C_1^{(k)}$ for all $k$.
\end{prop}

\begin{proof}
Consider the action of conjugation of the operator $A=Z(\theta)$ on the single qubit Pauli matrix~$X$, the action on Pauli~$Z$ is trivial due to the commutation of diagonal matrices. Consider the recursive construction of the matrices $A_p$ defined as: $A_p = A_{p-1} X A_{p-1}^{\dagger}$, where $A_0 = A$. Notice the following:
\begin{align*}
A_1 & = A_0XA_0^\dagger = 
\begin{pmatrix}
1 & 0 \\
0 & e^{i\pi\theta}
\end{pmatrix}
\begin{pmatrix}
0 & 1 \\
1 & 0
\end{pmatrix}
\begin{pmatrix}
1 & 0 \\
0 & e^{-i\pi\theta}
\end{pmatrix}
=
\begin{pmatrix}
0 & e^{-i\pi\theta} \\
e^{i\pi\theta} & 0
\end{pmatrix}, \\
A_2 &= A_1XA_1^\dagger = 
\begin{pmatrix}
0 & e^{-i\pi\theta} \\
e^{i\pi\theta} & 0
\end{pmatrix}
\begin{pmatrix}
0 & 1 \\
1 & 0
\end{pmatrix}
\begin{pmatrix}
0 & e^{-i\pi\theta} \\
e^{i\pi\theta} & 0
\end{pmatrix}
=
\begin{pmatrix}
0 & e^{-2i\pi\theta} \\
e^{2i\pi\theta} & 0
\end{pmatrix} ,\\
\vdots \\
A_p &= A_{p-1}XA_{p-1}^\dagger = 
\begin{pmatrix}
0 & e^{-2^p i\pi\theta} \\
e^{2^p i\pi\theta} & 0 
\end{pmatrix}.
\end{align*}
If $A \in \cal{C}_1^{(k+1)}$ for some $k \ge 0$, then by definition $A_1 \in \cal{C}_1^{(k)},\ A_2 \in \cal{C}_1^{(k-1)},\ \hdots, \ A_k \in \cal{C}^{(1)} = \cal{P}_1$. However, notice by the form of $A_k$ that $A_k = X \Leftrightarrow \theta = c /2^{k-1}$, $A_k = Y \Leftrightarrow \theta = c  /2^{k}$, and $A_k \ne Z \ \forall \ \theta$, where the angle~$\theta$ is in its most reduced form. 
\end{proof}

\begin{corollary}
Strongly transversal logical gates~$Z(\theta)^{\otimes n}$ on CSS~stabilizer codes must be composed of individual rotations that are an element of the Clifford hierarchy, that is~$Z(\theta) \in \cal{C}_1^{(k)}$, for some value of~$k$. Moreover, the logical gate that is implemented must also be an element of Clifford hierarchy on the logically encoded subspace.
\end{corollary}

\begin{proof}
The first statement follows from Propositions 1~and~2. The second statement follows from considering the action of the individual rotations on the logical states written out in their expansion in terms of the computational basis.
\end{proof} 

\subsection{Stabilizer codes}
\begin{prop}
A nontrivial qubit stabilizer codes can only have strongly transversal $Z$ rotations which are of the form $Z(a/2^k)$.  
\label{prop:StronglyTransversalStabilizer}
\end{prop}

Given a stabilizer code with a set of generators~$\{ G_i \}_{i=1}^k$,we can project onto the stabilizer codespace~$\mathcal{C}_{\mathcal{S}}$ of the code by applying the projection operator~$\prod_i (I + G_i)$ to a given state of the $n$-qubit Hilbert space, 
  
\begin{equation}\label{eq:logstab}
|\psi_L\rangle = \frac{1}{2^{k/2}} \prod_{i = 1}^k(I + G_{i})|0\rangle^{\otimes n}.
\end{equation}

It is also worth pointing out that we assume that the state~$\ket{0}^{\otimes n}$ is not orthogonal to the stabilizer codespace. This assumption can fail, however as there will always exist a state in the computational basis that is not orthogonal to~$\mathcal{C}_{\mathcal{S}}$, we make this assumption without loss of generality as the remainder of the proof would be identical by replacing~$\ket{0}^{\otimes n}$ with such a state. Before we begin the formal proof of Proposition~\ref{prop:StronglyTransversalStabilizer}, we will present a few useful results.

\begin{lemma}
\label{lem:StabilizerCodewords}
Given a set of $n$-qubit Pauli operators~$\langle G_i \rangle_{i=1}^{n-k}$ forming a stabilizer code~$\cal{\cal{S}}$, and logical Pauli operator~$X_{L,j}$, $Z_{L,j}$ for $1 \le j \le k$ satisfying $\left[X_{L,j}, Z_{L,l}\right] = \delta_{jl}$, then there exists a set of $2^k$ orthonormal states of the following form:
\begin{align}
\label{eq:StabilizerCodeword}
\ket{\psi_m} = \sum_l i^{a_{m,l}} \ket{m_l},
\end{align}
where $a_{m,l}$ is an integer and $m_l$ is an $n$-bit binary string (these states will form a basis for the logical state space). Moreover, two different states cannot share any elements in the computation basis expansion. More precisely, given $\ket{\psi_p}, \ \ket{\psi_q}$ such that $p \ne q$ then $\bk{p_s}{q_t} = 0 \ \forall \ s,t$.
\end{lemma}

\begin{proof}
There must exist at least one computational basis state that has non-zero overlap with the stabilizer codespace~$\cal{C}_{\cal{S}}$. Without loss of generality, we assume that~$\ket{0}^{\otimes n}$ is such a state. Then, the following state is a codestate of~$\cal{C}_{\cal{S}}$,
\begin{align*}
\ket{\phi} = \dfrac{1}{2^{(n-k)/2}} \prod_{i=1}^{n-k} (I+G_i) \ket{0}^{\otimes n} &= \dfrac{1}{2^{(n-k)/2}} \sum_i S_i \ket{0}^{\otimes n} \\
&= \dfrac{1}{2^{(n-k)/2}} \sum_i \ket{s_i},
\end{align*}
where we have defined the state~$\ket{s_i} = S_i \ket{0}^{\otimes n}$. Consider the action of two anti-commuting logical Pauli operators~$X_{L,1}$, $Z_{L,1}$, on the state~$\ket{\phi}$. We know that~$\ket{\phi}$ cannot be an eigenstate of both operators, as no state can be a joint eigenstate of two anti-commuting operators. Therefore, we can consider the following two cases: either $\ket{\phi}$ is an eigenstate of one of the operators, or~$\ket{\phi}$ is not an eigenstate of either operator. We shall consider the case of the former first.

Without loss of generality, assume that $Z_{L,1} \ket{\phi} = \ket{\phi} = \ket{\psi_1}$ and~$X_{L,1} \ket{\phi} = \ket{\psi_2} \ne \alpha\ket{\psi_1}$ (where~$\alpha$ is a global phase). Consider the action of~$X_{L,1}\ket{\phi}$:
\begin{align*}
X_{L,1} \ket{\phi} &= \dfrac{1}{2^{(n-k)/2}} \prod_{i=1}^{n-k} (I+G_i) X_{L,1} \ket{0}^{\otimes n} \\
&= \dfrac{1}{2^{(n-k)/2}} \sum_i S_i \ket{g_{X_{L,1}}},
\end{align*}
if $\ket{g_{X_{L,1}}} = \ket{s_j}$ for some~$j$ then after the action of the sum over stabilizer operators, the final state~$X_{L,1} \ket{\phi} = \alpha \ket{\phi}$ would be a contradiction. Therefore, the state~$\ket{g_{X_{L,1}}} $ must be a computational basis state that is not present in the expansion of~$\ket{\phi}$, and moreover, each element of the state~$\ket{\psi_2}$ must have zero overlap with the state~$\ket{\psi_1}$,
\begin{align*}
\ket{\psi_2} = X_{L,1} \ket{\phi} =  \dfrac{1}{2^{(n-k)/2}} \sum_i  \ket{g_{X_{L,1}} \oplus s_i}.
\end{align*}
Therefore, two states of the form of Equation~\ref{eq:StabilizerCodeword} have been constructed. Consider now the action of the next pair of anti-commuting logical Paulis~$X_{L,2}$, $Z_{L,2}$ on the state~$\ket{\psi_1}$. Again, since $\ket{\psi_1}$ cannot be a joint eigenstate of both operators, without loss of generality, assume $X_{L,2} \ket{\psi_1} = \ket{\psi_3} \ne \alpha \ket{\psi_1}$. Moreover, it must be that~$\ket{\psi_3} \ne \alpha \ket{\psi_2}$ or else the following would be true: $X_{L,1} X_{L,2} \ket{\psi_1} = \alpha X_{L,1} \ket{\psi_2} = \alpha \ket{\psi_1}$, which would imply that~$\ket{\psi_1}$ is an eigenstate of two anti-commuting operators, $Z_{L,1}$ and $X_{L,1}X_{L,2}$, which results in a contradiction. Therefore, we can express the state~$\ket{\psi_3}$ as follows:
\begin{align*}
\ket{\psi_3} = X_{L,2} \ket{\psi_1} &= \dfrac{1}{2^{(n-k)/2}} \prod_{i=1}^{n-k} (I+G_i) X_{L,2} \ket{0}^{\otimes n} \\
&=  \dfrac{1}{2^{(n-k)/2}} \sum_i  \ket{g_{X_{L,2}} \oplus s_i},
\end{align*}
where each state in the computational basis expansion must have zero overlap with the states~$\ket{\psi_1}$~and~$\ket{\psi_2}$. Finally. consider the action of the same anti-commuting pair on the state~$\ket{\psi_2}$. As will be shown below, it does not matter which we choose, and thus, without loss of generality, we assume it to be the state~$Z_{L,2}$. First note that if $Z_{L,2} \ket{\psi_2} = \alpha \ket{\psi_1}$, then $\ket{\psi_1}$ would be the joint eigenstate of two anti-commuting Paulis, $Z_{L,1}$ and~$X_{L,1} Z_{L,2}$, which is a contradiction. Moreover, if $Z_{L,2} \ket{\psi_2}  = \alpha \ket{\psi_3}$ then again $\ket{\psi_1}$ would be the joint eigenstate of two anti-commuting Paulis, $Z_{L,1}$ and~$X_{L,1} X_{L,2} Z_{L,2}$. Therefore, the state~$Z_{L,2} \ket{\psi_2}= \ket{\psi_4}$ must have zero overlap with the previous established states and can be expressed as follows:
\begin{align*}
\ket{\psi_4} = Z_{L,2} \ket{\psi_2} &= \dfrac{1}{2^{(n-k)/2}} \prod_{i=1}^{n-k} (I+G_i) Z_{L,2} X_{L,1} \ket{0}^{\otimes n} \\
&=  \dfrac{1}{2^{(n-k)/2}} \sum_i  \ket{g_{X_{L,1}} \oplus g_{Z_{L,2}} \oplus s_i}.
\end{align*}

Notice the form of~$\ket{\psi_3}$ and~$\ket{\psi_4}$. By taking the previous states~$\ket{\psi_1}$ and~$\ket{\psi_2}$ and a pair of non-commuting logical Paulis, for each state in the previous level, we can construct a new state by applying the logical Pauli for which it is not an eigenstate. One can continue the same constructive process for preparing states of the form of Equation~\ref{eq:StabilizerCodeword} by taking the $m$-th pair of anti-commuting logical operators and the $2^{m-1}$ previous constructed states, thereby creating another~$2^{m-1}$ set of orthogonal states, following similar constraints as laid out above. Applying this to all pairs of logical Pauli gates for the given code, $2^k$ basis states for the codespace can be constructed.

In the case when the state~$\ket{\phi}$ is not an eigenstate of either of the first two logical Pauli gates, $Z_{L,1}$ and $X_{L,1}$, the following modifications have to be made. Let~$\ket{\psi_1} = Z_{L,1} \ket{\phi}$ and $\ket{\psi_2} = X_{L,1} \ket{\phi}$. If $\ket{\psi_1} = \ket{\psi_2}$ then by redefining the logical Pauli $\tilde{Z_{L,1}} = X_{L,1} Z_{L,1}$ and $\ket{\tilde{\psi_1}} = \ket{\phi}$, we recover the original case where~$\tilde{Z_{L,1}} \ket{\tilde{\psi_1}} = \ket{\tilde{\psi_1}} = \ket{\phi} $ and $\ket{\psi_2} = X_{L,1} \ket{\phi}$. Therefore, the final case to consider is where~$\ket{\psi_1} \ne \alpha \ket{\psi_2}$. In this case, they must not have overlapping states in the computational basis, and their expansion can be written as follows:
\begin{align*}
\ket{\psi_1} = Z_{L,1} \ket{\phi} =  \dfrac{1}{2^{(n-k)/2}} \sum_i  \ket{g_{Z_{L,1}} \oplus s_i},\\
\ket{\psi_2} = X_{L,1} \ket{\phi} =  \dfrac{1}{2^{(n-k)/2}} \sum_i  \ket{g_{X_{L,1}} \oplus s_i}.
\end{align*}
Again, as in the previous case, consider the action of the pair of logical Pauli gates~$X_{L,2}$ and~$Z_{L,2}$ on the state $\ket{\psi_1}$. Without loss of generality, assume that $\ket{\psi_1}$ is not an eigenstate of~$X_{L,2}$. Unlike the previous case, it is now possible that~$X_{L,2} \ket{\psi_1} = \ket{\psi_2}$. However, if this holds, then redefining $\tilde{Z_{L,1}} = X_{L,1} Z_{L,1} X_{L,2}$ and $\ket{\tilde{\psi_1}} = \ket{\phi}$ we recover the original case with~$\tilde{Z_{L,1}} \ket{\tilde{\psi_1}} = \ket{\tilde{\psi_1}}$ and $X_{L,2} \ket{\tilde{\psi_1}} = \ket{\psi_3}$, as well as all redefined operators satisfying the appropriate commutation relations. Otherwise, we can conclude that~$X_{L,2} \ket{\psi_1} = \ket{\psi_3}$ and must be orthogonal to the two previous states as well as have the following form:
\begin{align*}
\ket{\psi_3} = X_{L,2} \ket{\psi_1} =  \dfrac{1}{2^{(n-k)/2}} \sum_i  \ket{g_{Z_{L,1}} \oplus g_{X_{L,2}} \oplus s_i}.
\end{align*}
Therefore, continuing in the same manner as in the previous case, we can construct the set of $2^k$ logical basis states of the form given by Equation~\ref{eq:StabilizerCodeword}.

\end{proof}

\begin{corollary}
\label{cor:StabilizerBasis}
Suppose $\cal{C}_{\cal{S}}$ is an $n$-qubit stabilizer containing~$k$ logical qubits. Given~$2^k$ states~$\ket{\varphi_m}\in \cal{C}_{\cal{S}}$ whose expansion in terms of the computational basis states are all non-overlapping, then these states must be of the form
\begin{align*}
\ket{\varphi_m} = \sum_l i^{a_{m,l}} \ket{m_l}.
\end{align*}
\end{corollary}

\begin{proof}
Since all $2^k$ states are elements of~$\cal{C}_{\cal{S}}$, they must be convex combinations of any basis chosen for~$\cal{C}_{\cal{S}}$. Choose the basis given by the states from Lemma~\ref{lem:StabilizerCodewords}. Then, if any of the~$\ket{\varphi_m}$ were a convex combination of states from such a basis, there must be at least one overlapping state relative to the individual states in its computational basis state expansion. Otherwise the dimension of the logical Hilbert space would be too small to fit all of these logical states.
\end{proof}

We know by Claim~\ref{lem:StabilizerCodewords}, that the computational basis state expansion of~$\ket{1_L}$ will be a sum of states such that each state differs from those in the representation of~$\ket{0_L}$. Moreover, the gate~$Z(\theta)^{\otimes n}$ will preserve all of these basis states, potentially introducing relative phases between the elements of the sum, however by Corollary~\ref{cor:StabilizerBasis} the resulting states must also form a basis for the stabilizer code, and in particular for the case of an automorphism the states must form the same logical basis. We now proceed with the proof of Proposition~\ref{prop:StronglyTransversalStabilizer}. 

\begin{proof}
We can represent a general Pauli string as a binary matrix using $\{I\rightarrow00, X\rightarrow10, Y\rightarrow11, Z\rightarrow01\}$. We will write an $n$-qubit Pauli string as a 2$n$-bit string $f=(g|h)$. Here we have separated the string into the two substrings of $n$-bits (an $X$~$(g)$ and $Z$~$(h)$ substring). We can express the expansion of~$|0_L\rangle$ and~$|1_L\rangle$ in terms of binary strings as

\begin{align}
\ket{0_L} = \prod_{i}(I + G_i)|0\rangle^{\otimes n} &=
|g_0\rangle + \ket{g_{Z_L}} + \sum_{i_1}(|g_{i_1}\rangle + \ket{g_{Z_L} \oplus g_{i_1}}) + \sum_{i_1<i_2}(|g_{i_1}\oplus g_{i_2}\rangle + \ket{g_{Z_L} \oplus g_{i_1} \oplus g_{i_2}})\\
& \qquad +...+ \sum_{i_1<i_2<...<i_{|G|}}(|g_{i_1}\oplus...\oplus g_{i_{|G|}}\rangle + \ket{g_{Z_L} \oplus g_{i_1}\oplus...\oplus g_{i_{|G|}}})
\end{align}

and 

\begin{align}
X_L \ket{0_L} &=
|g_{X_L} \rangle + \ket{g_{X_L} \oplus g_{Z_L}} + \sum_{i_1}(|g_{X_L} \oplus g_{i_1}\rangle + \ket{g_{X_L} \oplus g_{Z_L} \oplus g_{i_1}}) \\
& \qquad +...+ \sum_{i_1<i_2<...<i_{|G|}}(|g_{X_L} \oplus g_{i_1}\oplus...\oplus g_{i_{|G|}}\rangle + \ket{g_{X_L} \oplus g_{Z_L} \oplus g_{i_1}\oplus...\oplus g_{i_{|G|}}})
\end{align}

Here $g_0$ is the ``all-zeros'' string, $g_i(g_{X_L})$ is a binary string corresponding to the location of the $X$~Pauli operators in the set $G_i(X_L)$, and $\oplus$ is bitwise XOR. Rows of $G_{X_i}$ and $X_L$ can be expressed as binary strings with the association $\{I\rightarrow00, X\rightarrow10, Y\rightarrow11, Z\rightarrow01\}$, and all strings have length $2n$.

The effect of applying a $Z(\theta)^{\otimes n}$ rotation to a string $f=(g|h)$ will be
\begin{equation}
Z(\theta)^{\otimes n}|g\rangle = e^{i\pi\theta |g_X|}|g\rangle.
\end{equation}

For the CSS codes, we assumed that $X_L(Z_L)$ consisted of single qubit unitaries $X$ and $I$ ($Z$ and $I$). In this case we make no such assumptions. 

\begin{align*}
 \theta|g_{i_1}| &= 0 \bmod 2\\
 \theta|g_{i_1}\oplus g_{i_2}| &= 0 \bmod 2\\
 \vdots& \\
 \theta|g_{i_1}\oplus...\oplus g_{i_{|G|}}| &= 0 \bmod 2\\
 \\
 \theta|g_{{X_L}}| &\ne 0 \bmod 2\\
 \theta|g_{{X_L}}\oplus g_{i_1}| &\ne 0 \bmod 2\\
 \vdots& \\
\theta|g_{{X_L}}\oplus g_{i_1}\oplus...\oplus g_{i_{|G|}}| &\ne 0 \bmod 2\\
 \\
 \theta|g_{{Z_L}}| &= 0 \bmod 2\\
& \forall 0<i_1<i_2<...<i_{|G|}\le|G|
\end{align*}

The additional requirement is from $[Z_L, Z(\theta)^{\otimes n}] = 0$. Otherwise $Z(\theta)^{\otimes n}|0_L\rangle=|0_L\rangle \ne Z(\theta)^{\otimes n}Z_L|0_L\rangle$. 

These constraints are the same as before (actually slightly more constraining) and the proof carries through analogously. 

\end{proof}

\subsection{Relaxing strong transversality}

\begin{prop}
A nontrivial stabilizer code can only have transversal $Z$ rotations which are of the form $Z(a/2^k)$. 
\label{prop:GeneralStabilizer}
\end{prop}

First, notice that if the transversal operator includes the identity anywhere, it will have no effect on that qubit and therefore, we can formulate the overlap conditions on a new code with that qubit removed. Unlike puncturing a code we are not actually removing the qubit from the code; it is simply not included in the overlap conditions. In what follows, we will assume this process has been implemented, and no identity operators remain. We can do this without any difficulties since our overlap conditions make no use of the commuting properties of stabilizer generators.

To prove this more general case we will introduce a new tool; the {\it decompression lemma}. 

\begin{lemma}
If an $[[n,k,d]]$ code exists with a transversal $Z_{T}(\theta)=Z(\theta_1)\otimes Z(\theta_2)\otimes...\otimes Z(m\theta_n)$ gate, then there exists an $[[n+m-1,k,2]]$ code with a transversal $Z_{T}'(\theta)=Z(\theta_1)\otimes Z(\theta_2)\otimes...\otimes (Z(\theta_n)^{\otimes m})$ gate. 
\end{lemma}

This lemma is quite useful, yet nearly trivial. As we have shown, the overlap conditions on $X$ stabilizer generators and logical operators completely determine whether a specific $Z$ rotation can be implemented transversally. Then, if a code admits a transversal operation $Z_{T}(\theta) = Z(\theta_1)\otimes Z(\theta_2)\otimes...\otimes Z(m\theta_n)$, this code's $X$ stabilizer generators and logical operators clearly satisfy the overlap conditions for the transversal $Z$ operator. Now, if we take the last column of the check matrix and repeat it $m$ times, we have a new code which has distance two since a repeated column in the check matrix creates a weight two logical operator. It is easy to see that $Z_{T}'(\theta) = Z(\theta_1)\otimes Z(\theta_2)\otimes...\otimes (Z(\theta_n)^{\otimes m})$ satisfies the same overlap conditions on the new code that $Z_T$ satisfied for the original code, and it follows that $Z_{T}'(\theta)$ implements the same logical operation as $Z_{T}(\theta)$. Here we have not specified the $Z$ stabilizer generators and it should be noted that in the new code obtained after applying the decompression lemma, there will be $m-1-n$ new $Z$ stabilizer generators.

In this case we have a transversal gate 
\begin{equation}
Z_T(\theta) := Z(\theta_1)\otimes Z(\theta_2)\otimes...\otimes Z(\theta_n).
\end{equation}

In the appendix (see Section~\ref{app:rational}), we prove that irrational angles must cancel and therefore add nothing. We can therefore assume that $Z(\theta_i)$ is rational. 

Therefore, we have a transversal gate of the form
\begin{equation}
Z_T(\theta) := Z(p_1/q_1)\otimes Z(p_2/q_2)\otimes...\otimes Z(p_n/q_n).
\end{equation}

We can find the least common denominator $q$ of $q_1,...,q_n$ and express this as 
\begin{equation}
Z_T(\theta) := Z(p'_1/q)\otimes Z(p'_2/q)\otimes...\otimes Z(p'_n/q).
\end{equation}

We can also use $Z(2+p/q) = Z(p/q)$ to put each $p'_i/q_i$ in $[0,2)$. We also assume that $Z(p/q)\ne I$ since we could ignore this operator and the qubit it acts upon, as they do not affect the overlap conditions.  

Now, we repeatedly apply the decompression lemma until we have an $[[\sum_i p_i,k,2]]$ code with a transversal gate
\begin{equation}
Z_T(\theta) := Z(1/q)\otimes...\otimes Z(1/q).
\end{equation}

We have now reduced these more general gates to strongly transversal gates, and the proof follows as before.

\subsection{Classification of all single qubit logical gates}

Recall that Zeng~\textit{et al.} showed that all single-qubit logical transversal gates for a stabilizer code must have the form~\cite{Zeng:2007a}:

\begin{equation}
U = L\left(\bigotimes_{j=1}^{n}\mbox{diag}(1,e^{i\pi\theta_j})\right)R^{\dagger}P_{\pi},
\end{equation}
where $P_{\pi}$ is a permutation matrix of the physical qubits while $R$ and~$L$ are transversal single-qubit Clifford operators. Let $D = \bigotimes_{j=1}^{n}\mbox{diag}(1,e^{i\pi\theta_j})$ and note that, in the case where $LR^{\dagger}P_{\pi}$ is an automorphism, \ie~they preserve the stabilizer codespace, then $D$ must also be an automorphism and the gate restrictions from Proposition~\ref{prop:StronglyTransversalStabilizer} must hold.

Since $P_{\pi}$ permutes the physical qubits of the original stabilizer code, and~$R^{\dagger}$ is a transversal Clifford operation after the application of these two gates, a state that was originally in codespace of a stabilizer code~$\cal{S}$ must also be a a codespace of a stabilizer code~$\cal{S}'$ (which could potentially be the same stabilizer code).

\begin{prop}
Given two nontrivial $n$-qubit stabilizer codes~$\cal{S}$ and~$\cal{T}$ consisting of~$k$ logical qubits, strongly transversal $Z$ rotations which map $\cal{S} \longrightarrow \cal{T}$ (and possibly apply a logical unitary in the process) must be of the form $Z(a/2^k)$.  
\label{prop:StronglyTransversalStabilizerConv}
\end{prop}

\begin{proof}
Let $\{ \ket{\psi_m}_{m=1}^{2^k} \}$ form a logical basis set for the stabilizer code and choose~$\cal{S}$ to be of the form outlined in Lemma~\ref{lem:StabilizerCodewords}. Then, given a transversal application of diagonal gates, the resulting set of states must also form a basis for a stabilizer code (in this case chosen to be~$\cal{T}$) such that each individual basis state will have the same expansion in terms of the computational basis states. However, the transversal application of diagonal gates may result in relative phases between the states; since the states must form a basis for the stabilizer code of the type given by Lemma~\ref{lem:StabilizerCodewords}, the relative phases must be powers of~$i$. Therefore, the transformed states read:
\begin{align*}
\ket{\psi_m} = \sum_l i^{a_{m,l}} \ket{m_l} \xrightarrow{D} \ket{\varphi_m} = e^{i \pi \phi_m}\sum_l i^{a_{m,l}} i^{c_{m,l}} \ket{m_l} .
\end{align*}
Therefore, repeating the action of the diagonal transversal gate 4 times must return the original set of basis states (with the possible introduction of a phase).
\begin{align*}
\ket{\psi_m} = \sum_l i^{a_{m,l}} \ket{m_l} \xrightarrow{D^4} e^{i 4 \pi \phi_m} \ket{\psi_m} = e^{i 4 \pi \phi_m}\sum_l i^{a_{m,l}} \ket{m_l} .
\end{align*}
We are now back to the original case of classifying transversal diagonal gates for logical gates returning to the same codespace, which we have already classified to be rotations of the form~$Z(a/2^k)$. Therefore, we are similarly restricted for the case of logical mappings between stabilizer codes.

\end{proof}

\section{Multi-block gates}

Consider now the case of $r$~blocks of the same error correcting code~$Q=[[n,k,d]]$. Zeng~\etal classified the set of gates that can be transversal across these codeblocks. Namely, if $U$ is a transversal gate on $Q^{\otimes r}$, then for each $j \in [n]$ either $U_j \in \mathcal{L}_r$ or $U_j = L_1 V L_2$, where $L_1, L_2 \in \mathcal{L}_1^{\otimes r}$ are local Clifford gates, and $V$ keeps the linear span of the group elements of~$\langle \pm Z_j^{(i)}, i\in [r] \rangle$.

This work focuses on the gates~$V$, which must be diagonal in order to preserve the span of the group of $Z$ operators across qubits at a fixed~$i$.

\subsection{Strong transversality for two-qubit logical gates}

Consider first the implementation of a logical diagonal gate in the case of two codeblocks, where the logical gate is implemented by using a strongly transversal gate. That is, consider the implementation of the diagonal two-qubit logical gate by applying a given two-qubit gate,
$U = \sum_j e^{i \pi \theta_j }$ transversally~$U^{\otimes n}$ among the corresponding pair of qubits between the codeblocks. The desired logical gate to be implemented has the form
\begin{align*}
U_L = \sum_j e^{i \pi \omega_j },
\end{align*}
where the states~$\{ \ket{j}_L \}_j = \{ \ket{00}_L, \ket{01}_L, \ket{10}_L, \ket{11}_L \} $ are two-qubit logical states spanning the two codeblocks.

As in the single block case, the desired action of the logical gate on the logical states will impose a restriction on the form of the two-qubit physical gates that can be implemented in a strongly transversal manner. In the case of a quantum CSS code, the above logical gate description will have the following form:
\begin{align*}
U_L \ket{00}_L &= U^{\otimes n} \prod_i (I+G_{X_i}) \ket{0}^{\otimes n} \prod_j (I+G_{X_j}) \ket{0}^{\otimes n} \\
&= U^{\otimes n} \Big( \ket{g_0} + \sum_{i_1} \ket{g_{i_1}} + \sum_{i_1<i_2} \ket{g_{i_1} \oplus g_{i_2}} + \hdots + \sum_{i_1 < i_2 < \hdots < i_{|G_X|}} \ket{i_1 =\oplus \hdots \oplus i_{|G_X|}} \Big) \\
& \qquad \otimes \Big( \ket{g_0} + \sum_{j_1} \ket{g_{j_1}} + \sum_{j_1<j_2} \ket{g_{j_1} \oplus g_{j_2}} + \hdots + \sum_{j_1 < j_2 < \hdots < j_{|G_X|}} \ket{j_1 =\oplus \hdots \oplus j_{|G_X|}} \Big)\\
&=e^{i \pi \omega_{00}} \ket{00}_L.
\end{align*}

Note that each of the $4^{|G_X|}$~states in the summation of the $\ket{00}_L$ state are computational basis states and will not change under the action of~$U^{\otimes n}$ except for the possible addition of a phase. Therefore, in order to remain a codeword, all states in the expansion must have the same phase.

Without loss of generality, one can assume that the phase~$\theta_{00}=0$ (this maps to a global phase freedom in the logical gate~$U_L$). Consider now the phases introduced on all $2^{|G_X|}$ states in the expansion of the first qubit, along with the state $\ket{g_0}$ in the expansion of the second qubit. For clarity, we will list the state in the expansion, along with the corresponding condition imposed on its phase.

\begin{align*}
&\ket{g_0}\ket{g_0}:	& n\theta_{00} = \omega_0 &= 0 \qquad \text{mod } 2\\
&\ket{g_{i_1}}\ket{g_0}:	& |g_{i_1}| \theta_{10} + (n-|g_{i_1}|)\theta_{00} = |g_{i_1}| \theta_{10} &= 0 \qquad \text{mod } 2 \\
&\ket{g_{i_1}\oplus g_{i_2}} \ket{g_0}:	& |g_{i_1} \oplus g_{i_2}| \theta_{10} &= 0 \qquad \text{mod } 2 \\
& \vdots \\
&\ket{g_{i_1}\oplus \hdots \oplus g_{i_{|G_X|}}} \ket{g_0}:	& |g_{i_1} \oplus \hdots \oplus g_{i_{|G_X|}} | \theta_{10} &= 0 \qquad \text{mod } 2
\end{align*}

These conditions are equivalent to the conditions derived on a single codeblock and therefore $\theta_{10}$ is restricted to be an integer multiple of~$1/2^c$ (when paired with the appropriate restrictions on the logical phase~$\omega_{10}$ as given below). In a very similar manner, a set of constraints can be obtained for the angle~$\theta_{01}$ due to the symmetry of the two codes:
\begin{align*}
&\ket{g_0}\ket{g_{j_1}}:	& |g_{j_1}| \theta_{01} + (n-|g_{j_1}|)\theta_{00} = |g_{i_1}| \theta_{01} &= 0 \qquad \text{mod } 2 \\
&\ket{g_0}\ket{g_{j_1}\oplus g_{j_2}}:	& |g_{j_1} \oplus g_{j_2}| \theta_{01} &= 0 \qquad \text{mod } 2 \\
& \vdots \\
&\ket{g_0}\ket{g_{j_1}\oplus \hdots \oplus g_{j_{|G_X|}}}:	& |g_{j_1} \oplus \hdots \oplus g_{j_{|G_X|}} | \theta_{01} &= 0 \qquad \text{mod } 2 
\end{align*}
Therefore, the phase angle~$\theta_{01}$ will also be restricted to be an integer multiple of~$1/2^c$.

In order to obtain a restriction on the phase angle~$\theta_{11}$, higher order state vectors must be considered in both expansions of the logical~$\ket{0}_L$ states. Consider the full summation over states in the expansion of the first block, along with state vectors in the second block of the form~$\ket{g_{j_1}}$.

\begin{align*}
&\ket{g_{i_1}}\ket{g_{j_1}}:	& |g_{i_1}\wedge g_{j_1}| (\theta_{11} - \theta_{01} -\theta_{10})  + |g_{i_1}| \theta_{10} + |g_{j_1}| \theta_{01} &= 0 \qquad \text{mod } 2 \\
&\ket{g_{i_1}\oplus g_{i_2}}\ket{g_{j_1}}:	& |(g_{i_1} \oplus g_{i_2}) \wedge g_{j_1}| (\theta_{11} - \theta_{01} -\theta_{10})  + |g_{i_1} \oplus g_{i_2}| \theta_{10} + |g_{j_1}| \theta_{01} &= 0 \qquad \text{mod } 2 \\
& \vdots \\
&\ket{g_{i_1}\oplus \hdots \oplus g_{i_{|G_X|}}}\ket{g_{j_1}}:	& |(g_{i_1} \oplus \hdots \oplus g_{i_{|G_X|}}) \wedge g_{j_1}| (\theta_{11} - \theta_{01} -\theta_{10}) \qquad  \\
&&+ |g_{i_1} \oplus \hdots \oplus g_{i_{|G_X|}}| \theta_{10} + |g_{j_1}| \theta_{01} &= 0 \qquad \text{mod } 2 \\
\end{align*}

First notice that every term other than the first term in each of the conditions will be equal to zero (mod~2), as a result of the set of conditions imposed on the phase angles~$\theta_{01}$ and~$\theta_{10}$. Therefore, what remains are conditions on the phase difference~$\theta_{11}' = (\theta_{11} - \theta_{01} -\theta_{10})$, which is equivalent to a condition on~$\theta_{11}$. Moreover, consider the expansion of the direct sum as given by Equation~\ref{eq:BravyiHaah}:

\begin{align*}
|(g_{i_1} \oplus g_{i_2}) \wedge g_{j_1}| \theta_{11}' &= |(g_{i_1} \wedge g_{j_1}) \oplus (g_{i_2} \wedge g_{j_1})| \theta_{11}' \\
&= \Big(  |g_{i_1} \wedge g_{j_1}| + | g_{i_2} \wedge g_{j_1}| - 2  |g_{i_1} \wedge g_{i_2} \wedge g_{j_1}| \Big) \theta_{11}' &= 0 \\
\Rightarrow  2  |g_{i_1} \wedge g_{i_2} \wedge g_{j_1}|  \theta_{11}'  = 0,
\end{align*}
where the implication in the final line is due to~$|(g_{i} \wedge g_{j_1})|\theta_{11}'=0$, for all~$i$ from the first set of constraints of the state vector~$\ket{g_{i_1}}\ket{g_{j_1}}$. Similarly,

\begin{align*}
|(g_{i_1} \oplus g_{i_2} \oplus g_{i_3}) \wedge g_{j_1}| \theta_{11}' &= |(g_{i_1} \wedge g_{j_1}) \oplus (g_{i_2} \wedge g_{j_1}) \oplus (g_{i_3} \wedge g_{j_1}) | \theta_{11}' \\
&= \Big(  |g_{i_1} \wedge g_{j_1}| + | g_{i_2} \wedge g_{j_1}| + |g_{i_3} \wedge g_{j_1}|  \\
& \qquad - 2  |g_{i_1} \wedge g_{i_2} \wedge g_{j_1}|  - 2  |g_{i_1} \wedge g_{i_3} \wedge g_{j_1}|  - 2  |g_{i_2} \wedge g_{i_3} \wedge g_{j_1}| \\
&\qquad \qquad +4 |g_{i_1} \wedge g_{i_2} \wedge g_{i_3} \wedge g_{j_1}|  \Big) \theta_{11}' = 0 \\
\Rightarrow  4  |g_{i_1} \wedge g_{i_2} \wedge g_{i_3} \wedge g_{j_1}|  \theta_{11}'  = 0.
\end{align*}
The final implication is a consequence of the above condition on~$2  |g_{i_1} \wedge g_{i_2} \wedge g_{j_1}|$. The same procedure will follow for all states in the expansion, and conditions on~$\theta_{11}'$ can thus be modified as:
\begin{align*}
&\ket{g_{i_1}}\ket{g_{j_1}}:	& |g_{i_1}\wedge g_{j_1}| \theta_{11}' &= 0 \qquad \text{mod } 2 \\
&\ket{g_{i_1}\oplus g_{i_2}}\ket{g_{j_1}}:	& 2 |g_{i_1} \wedge g_{i_2} \wedge g_{j_1}| \theta_{11}' &= 0 \qquad \text{mod } 2 \\
&\ket{g_{i_1}\oplus g_{i_2} \oplus g_{i_3}}\ket{g_{j_1}}:	& 4 |g_{i_1} \wedge g_{i_2} \wedge g_{i_3} \wedge g_{j_1}| \theta_{11}' &= 0 \qquad \text{mod } 2 \\
& \vdots \\
&\ket{g_{i_1}\oplus \hdots \oplus g_{i_{|G_X|}}}\ket{g_{j_1}}:	& 2^{|G_X| - 1} |g_{i_1} \wedge \hdots \wedge g_{i_{|G_X|}} \wedge g_{j_1}| \theta_{11}'  &= 0 \qquad \text{mod } 2.
\end{align*}

Given that the two codebocks are encoded in the same quantum error correcting code, these conditions are a modified version of the conditions derived in the single block case, where an extra factor of~2 is present in all of the constraints. This extra factor of 2 will have a consequence on the type of logical gates that can be implemented transversally and will limit the 2-qubit gates to reside in the same level of the Clifford hierarchy as the 1-qubit gates that can be implemented for a given code.

Finally, consider the action of the strongly transversal gate on the logical states~$\ket{01}_L$, $\ket{10}_L$, and~$\ket{11}$ when performing a logical~$X_L$ on the appropriate qubit(s). The resulting set of conditions impose a restriction on the logical phases~$\omega_{01}$,~$\omega_{10}$, and~$\omega_{11}$. For the logical state~$\ket{01}_L$ the conditions are:
\begin{align*}
&\ket{g_0}\ket{g_{L}}:	&  |g_L| \theta_{01} &= \omega_{01} \qquad \text{mod } 2 \\
&\ket{g_0}\ket{g_{j_1}\oplus g_L}:	& |g_{j_1} \oplus g_L| \theta_{01} &= \omega_{01} \qquad \text{mod } 2 \\
& \vdots \\
&\ket{g_0}\ket{g_{j_1}\oplus \hdots \oplus g_{j_{|G_X|}} \oplus g_L}:	& |g_{j_1} \oplus \hdots \oplus g_{j_{|G_X|}} \oplus g_L | \theta_{01} &= \omega_{01} \qquad \text{mod } 2.
\end{align*}
Therefore these restrictions, along with the conditions for the phase~$\theta_{01}$, will impose the restriction of the form of phases that can be applied, as shown in the single block case. In the exact same manner, restriction on the phases~$\theta_{10}$ and~$\omega_{10}$ are obtained. Finally, in order to obtain restrictions on the phase~$\omega_{11}$, consider the following:
\begin{align*}
&\ket{g_L}\ket{g_{L}}:	&  |g_L| \theta_{11} &= \omega_{11} \qquad \text{mod } 2 \\
&\ket{g_L}\ket{g_{j_1}\oplus g_L}:	& |g_{j_1} \oplus g_L| \theta_{11} &= \omega_{11} \qquad \text{mod } 2 \\
& \vdots \\
&\ket{g_L}\ket{g_{j_1}\oplus \hdots \oplus g_{j_{|G_X|}} \oplus g_L}:	& |g_{j_1} \oplus \hdots \oplus g_{j_{|G_X|}} \oplus g_L | \theta_{11} &= \omega_{11} \qquad \text{mod } 2,
\end{align*}
there conditions are in fact exactly the same as the overlap between the $g_L$~string in both states has to be the same since the two codes are encoded into the same codeblock. Therefore, the exact same overlap conditions are obtained for the phases~$\theta_{11}$ and~$\omega_{11}$. 

These set of conditions result in the following Theorem for two-qubit transversal diagonal gates.

\begin{theorem}\label{thm:multiq}
Given a quantum error correcting code that can implement the logical gate~$Z_L(1/2^k)$ by applying a transversal~$Z(1/2^k)^{\otimes n}$ on the underlying physical qubits yet cannot implement the gate~$Z_L(1/2^{k+1})$ due to code constraints. Then, the set of two-qubit diagonal gates $U=\sum_j e^{i\pi \theta_j}$ that can implement a logical two-qubit operation by applying such gates transversally~$U^{\otimes n}$ will be restricted to the angles (up to a global phase freedom),
\begin{align*}
\theta_{00} &= 0, \\
\theta_{01} &= a/2^k, \\
\theta_{10} &= b/2^k, \\
\theta_{11} &= a/2^k + b/2^k + c/2^{k-1},
\end{align*}
where $a,\ b$ and~$c$ are arbitrary integers. The resulting two-qubit logical gate will have the form (up to Clifford arbitrary Clifford gates),
\begin{align*}
U_L = \begin{pmatrix}
1 & 0 & 0 & 0 \\
0 & e^{i \pi \alpha/2^{k-1}} & 0 & 0 \\
0 & 0 & e^{i \pi \beta/2^{k}} & 0 \\
0 & 0 & 0 & e^{i \pi \gamma/2^{k}}
\end{pmatrix},
\end{align*}
where $\beta = 0 \Leftrightarrow \gamma = 0$. Moreover, the two-qubit logical gate is at the same level in the Clifford hierarchy as the single-qubit logical gate that transversal single qubit diagonal gates can implement, \ie~$Z_L(1/2^k) \in \cal{C}_1^{(k+1)}$ and $U_L \in \cal{C}_2^{(k+1)}$.
\end{theorem}

\begin{proof}
The first result of the proof is proved in the above section by the resulting constraints on the angles that can be implemented transversally. More specifically, since $Z_L(1/2^k)$ can be implemented transversally, we know that the generators of stabilizer of the code must satisfy,
\begin{align*}
|g_{i_1}| &= 0 \mod 2^k \\
2|g_{i_1} \wedge g_{i_2}| &= 0 \mod 2^k \\
&\vdots \\
2^{k-1}|g_{i_1} \wedge \hdots \wedge g_{i_k}| &= 0 \mod 2^k,
\end{align*}
for all choices of valid indices~$\{i_1,\hdots i_k\}$. Moreover, since the code is constrained to not be able to implement~$Z(1/2^{k+1})$ transversally, we know the following must be true for some choice of indices~$\{ \mu_1, \hdots , \mu_{k+1}\}$,
\begin{align*}
2^k |g_{\mu_1} \wedge \hdots \wedge g_{\mu_{k+1}}| \ne 0 \mod 2^{k+1}.
\end{align*}
Therefore, given the resulting set of constraints on the angle difference~$\theta_{11}' = \theta_{11}-\theta_{01} - \theta_{10}$,
\begin{align*}
|g_{i_1} \wedge g_{j}| &= 0 \mod q \\
2|g_{i_1} \wedge g_{i_2}\wedge g_{j}| &= 0 \mod q \\
&\vdots \\
2^{k-1}|g_{i_1} \wedge \hdots \wedge g_{i_k} \wedge g_{j}| &= 0 \mod q,
\end{align*}
where $1/q$ is the desired angular rotation. The above conditions will not be able to be satisfied for $q = 2^k$ as for the indices $\{ \mu_1, \cdots, \mu_{k+1} \}$ the following would lead to a contradiction with the final condition:
\begin{align*}
2^k |g_{\mu_1} \wedge \hdots \wedge g_{\mu_{k+1}}| \ne 0 \mod 2^{k+1} \Longrightarrow 
2^{k-1} |g_{\mu_1} \wedge \hdots \wedge g_{\mu_{k+1}}| \ne 0 \mod 2^{k}.
\end{align*}
Conversely, we know that we can satisfy the above equations for $q=2^{k-1}$ by using an implication from the single-qubit conditions that must be satisfied for all indices~$\{i_1,\hdots i_k\}$,
\begin{align*}
2^l |g_{i_1} \wedge \hdots \wedge g_{i_{l+1}}| = 0 \mod 2^{k} \Longrightarrow 
2^{l-1} |g_{i_1} \wedge \hdots \wedge g_{i_{l+1}}| = 0 \mod 2^{k-1}.
\end{align*}
Therefore, combining the results for the single qubit gates on the same restrictions and the above observations we know the angle~$\theta_{11}'$ is restricted to have the form~$c/2^{k-1}$. Since the angles~$\theta_{01}$ and~$\theta_{10}$ satisfy the same restrictions as the single qubit block case, this completes the proof of the first claim of Theorem~\ref{thm:multiq} regarding the allowable angles which a multi-qubit diagonal gate can implementing a logical multi-qubit gate via a transversal application of the chosen gates. In order to obtain a description of the two-qubit logical gate, we can consider the set of equations that provide the restrictions on the allowable angles in order to obtain an explicit description of the logical angle that can be applied.
\begin{align*}
U^{\otimes n} \ket{00}_L &= e^{i \pi n\theta_{00}} \ket{00}_L = \ket{00} \\
U^{\otimes n} \ket{01}_L &= e^{i \pi |g_i| \theta_{01}} \ket{01}_L = e^{i \pi a |g_i|/2^{k}}  \ket{01}_L = e^{i \pi \alpha/2^{k}}  \ket{01}_L \\
U^{\otimes n} \ket{10}_L &= e^{i \pi |g_i| \theta_{10}} \ket{10}_L = e^{i \pi b |g_i|/2^{k}}  \ket{10}_L = e^{i \pi \beta /2^{k}}  \ket{10}_L \\
U^{\otimes n} \ket{11}_L &= e^{i \pi |g_i \wedge g_j | \theta_{11}} \ket{01}_L = e^{i \pi  |g_i| (\theta_{01} + \theta_{10} + \theta_{11}')}  \ket{11}_L = e^{i \pi (a + b +2c) |g_i| /2^{k}}  \ket{11}_L = e^{i \pi (\alpha+ \beta +2\eta) /2^{k}}  \ket{11}_L
\end{align*}
The above equations must hold for any choice of the weight of the individual (or pairs) of stabilizer generators and in the last equation we have chosen $g_j$ to equal $g_i$ as we have a freedom over which $j$ we choose. We have also introduced the integers $\alpha = a |g_i|$, $\beta = b |g_i|$,~and $\eta = c |g_i|$. Consider the angle that is applied by the logical operation to the state~$\ket{11}_L$ in more detail, $(\alpha + \beta + 2\eta)/2^k$. If both~$\alpha$ and $\beta$ are odd, then the overall angle will be of the form~$\gamma/2^{k-1}$. If either $\alpha$ or $\beta$ are even (or zero), but not both, then the angle will have the form~$\gamma/2^k$; however this would mean that the other angle could then be expressed in its most reduced form as~$\alpha'/2^{k-1}$. Finally, if both are even, it follows that all these angles can be reduced and shown to be proportional to~$1/2^{k-1}$. Therefore, up to a relabelling of logical basis states (which can be achieved using either a logical~$X_L$ or $CNOT$ gate), the two-qubit logical gate can be expressed in the form
\begin{align*}
U_L = \begin{pmatrix}
1 & 0 & 0 & 0 \\
0 & e^{i \pi \alpha/2^{k}} & 0 & 0 \\
0 & 0 & e^{i \pi \beta/2^{k}} & 0 \\
0 & 0 & 0 & e^{i \pi \gamma/2^{k-1}}
\end{pmatrix},
\end{align*}
where $\alpha = 0 \Leftrightarrow \beta = 0$, such that either both phases are zero or the two phases are proportional to~$1/2^{k}$ (in the case when only one is zero, then the gate will be a product of a single logical qubit rotation proportional to~$Z(1/2^{k})$ and a controlled-$Z(1/2^{k-1})$ which are both in~$\cal{C}_2^{k}$). To prove the final statement of Theorem \ref{thm:multiq} we must show that the above gate is contained within the Clifford hierarchy at the $(k+1)$-th level of the two-qubit Clifford hierarchy, $U_L \in \cal{C}_2^{k+1}$.

We shall prove this by induction. Begin with the case $k=1$, and without loss of generality, assume $\alpha < \beta$, that both are not equal to zero, and that the angles are written in their most reduced form (if both $\alpha = \beta = 0$. Now the proof of the base case is trivial, as it becomes a controlled-$Z$ gate which is clearly in~$\cal{C}_2^{(2)}$, a Clifford gate). We can rewrite the logical gate as
\begin{align*}
U_L &= \begin{pmatrix}
1 & 0 & 0 & 0 \\
0 & e^{i \pi \alpha/2^{k}} & 0 & 0 \\
0 & 0 & e^{i \pi \beta/2^{k}} & 0 \\
0 & 0 & 0 & e^{i \pi \gamma/2^{k-1}}
\end{pmatrix} = 
\begin{pmatrix}
1 & 0 & 0 & 0 \\
0 & e^{i \pi \alpha/2^{k}} & 0 & 0 \\
0 & 0 & e^{i \pi \alpha/2^{k}} & 0 \\
0 & 0 & 0 & 1
\end{pmatrix} 
\begin{pmatrix}
1 & 0 & 0 & 0 \\
0 & 1 & 0 & 0 \\
0 & 0 & e^{i \pi (\beta-\alpha)/2^{k}} & 0 \\
0 & 0 & 0 & e^{i \pi \gamma/2^{k-1}}
\end{pmatrix} \\
&= 
\begin{pmatrix}
1 & 0 & 0 & 0 \\
0 & e^{i \pi \alpha/2^{k}} & 0 & 0 \\
0 & 0 & e^{i \pi \alpha/2^{k}} & 0 \\
0 & 0 & 0 & 1
\end{pmatrix} 
\begin{pmatrix}
1 & 0 & 0 & 0 \\
0 & 1 & 0 & 0 \\
0 & 0 & e^{i \pi (\beta-\alpha)/2^{k}} & 0 \\
0 & 0 & 0 & e^{i \pi (\beta-\alpha)/2^{k}}
\end{pmatrix}
\begin{pmatrix}
1 & 0 & 0 & 0 \\
0 & 1 & 0 & 0 \\
0 & 0 & 1& 0 \\
0 & 0 & 0 & e^{i \pi (2\gamma + \alpha - \beta)/2^{k}}
\end{pmatrix}.
\end{align*}
The above sequence of unitaries can be expressed as the following sequence of gates:
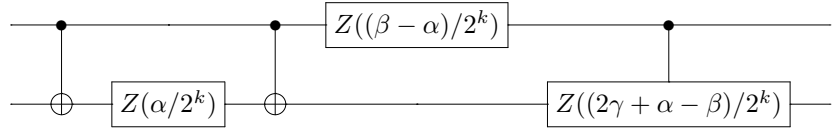
\begin{align*}
\Qcircuit @C=1.5em @R=1.3em {
& \ctrl{1}	& \qw				& \ctrl{1}	& \gate{Z((\beta-\alpha)/2^k)}	& \ctrl{1} & \qw \\
& \targ	& \gate{Z(\alpha/2^k)}	&\targ	& \qw		& \gate{Z( (2\gamma + \alpha - \beta)/2^k) } & \qw.
}
\end{align*}
In this case, $k=1$, and all of the single qubit gates are achieved by repeated action of the Clifford phase gate~$S = \text{diag}(1,i)$. The two-qubit coupling gate is actually the application of a controlled-$Z$ gate since $\alpha \ne 0$, $\beta \ne 0$, and both are odd; therefore, their difference is even and the gate can be expressed in the form~$Z(\zeta/2)$. Since all of these gates are in~$\cal{C}_2^{(2)}$ and the first two levels of the Clifford hierarchy form a group, the resulting composition is an element of~$\cal{C}_2^{(2)}$.

Assume the claim holds for $k-1$; we will now show that it holds for $k$. By definition, if $U_L \in \cal{C}_2^{(2)}$ it must map any two-qubit Pauli to an element in~$\cal{C}_2^{(k-1)}$ when conjugating by~$U_L$. We need only consider the action of~$U_L$ on the Pauli~$X$ elements, as the action on Pauli-$Z$ is trivial since diagonal gates commute. Consider the following:
\begin{align*}
U_L (X \otimes I) U_L^{\dagger} &= 
\begin{pmatrix}
1 & 0 & 0 & 0 \\
0 & e^{i \pi \alpha/2^{k}} & 0 & 0 \\
0 & 0 & e^{i \pi \beta/2^{k}} & 0 \\
0 & 0 & 0 & e^{i \pi \gamma/2^{k-1}}
\end{pmatrix}
\begin{pmatrix}
0 & 0 & 1 & 0 \\
0 & 0 & 0 & 1 \\
1 & 0 & 0 & 0 \\
0 & 1 & 0 & 0
\end{pmatrix}
\begin{pmatrix}
1 & 0 & 0 & 0 \\
0 & e^{-i \pi \alpha/2^{k}} & 0 & 0 \\
0 & 0 & e^{-i \pi \beta/2^{k}} & 0 \\
0 & 0 & 0 & e^{-i \pi \gamma/2^{k-1}}
\end{pmatrix} \\
&=
\begin{pmatrix}
0 & 0 & e^{-i \pi \beta/2^{k}} & 0 \\
0 & 0 & 0 & e^{i \pi (\alpha-2\gamma)/2^{k}} \\
e^{i \pi \beta/2^{k}}  & 0 & 0 & 0 \\
0 & e^{-i \pi (\alpha-2\gamma)/2^{k}} & 0 & 0
\end{pmatrix} = A.
\end{align*}
Through the action of CNOT gates, we can map the above operator to the following: 
\begin{align*}
e^{-i\pi \beta/2^k}
\begin{pmatrix}
1 & 0 & 0 & 0 \\
0 & e^{i \pi (\alpha + \beta -2\gamma)/2^{k}} & 0 & 0 \\
0 & 0 & e^{i \pi 2\beta/2^{k}} & 0 \\
0 & 0 & 0 & e^{-i \pi (\alpha- \beta -2\gamma)/2^{k}} 
\end{pmatrix}.
\end{align*}
Note that the left or right action of any Clifford gate will not change the level of an element in the Clifford hierarchy, as proven in Prop. 3 in Ref~\cite{Zeng2008}. Therefore we can show that the above gate is in~$\cal{C}_2^{(k-1)}$ which is equivalent to showing that~$A \in \cal{C}_2^{(k-1)}$. We know we can write the integers~$\alpha$ and~$\beta$ as $\alpha = 2k_\alpha+ 1$ and $2k_\beta +1$. Consider the following angular expressions:
\begin{align*}
\dfrac{\alpha + \beta - 2\gamma}{2^k} &= \dfrac{2(k_\alpha + k_\beta) - 2\gamma + 2}{2^k} = \dfrac{(k_\alpha + k_\beta) - (\gamma -1 )}{2^{k-1}}, \\ 
\dfrac{\alpha - \beta - 2\gamma}{2^k} &= \dfrac{2(k_\alpha - k_\beta) - 2\gamma }{2^k} = \dfrac{(k_\alpha - k_\beta) - \gamma}{2^{k-1}}. 
\end{align*}
Since $(k_\alpha + k_\beta)$ is even, if and only if $(k_\alpha - k_\beta)$ is even, one of the numerators in the final expression will be even, and as such, one of the above angles will necessarily be of the form~$1/2^{k-2}$. Therefore, up to a logical Clifford operation (which preserves the level of the Clifford hierarchy), the gate~$A$ will have the form:
\begin{align*}
\begin{pmatrix}
1 & 0 & 0 & 0 \\
0 & e^{i \pi \alpha'/2^{k-1}} & 0 & 0 \\
0 & 0 & e^{i \pi \beta'/2^{k-1}} & 0 \\
0 & 0 & 0 & e^{i \pi \gamma'/2^{k-2}} 
\end{pmatrix},
\end{align*}
which by the induction hypothesis is an element of~$\cal{C}_2^{(k-1)}$. Finally, we must show the same property for the following mapping:
\begin{align*}
U_L (I \otimes X) U_L^{\dagger} &= 
\begin{pmatrix}
1 & 0 & 0 & 0 \\
0 & e^{i \pi \alpha/2^{k}} & 0 & 0 \\
0 & 0 & e^{i \pi \beta/2^{k}} & 0 \\
0 & 0 & 0 & e^{i \pi \gamma/2^{k-1}}
\end{pmatrix}
\begin{pmatrix}
0 & 1 & 0 & 0 \\
1 & 0 & 0 & 0 \\
0 & 0 & 0 & 1 \\
0 & 0 & 1 & 0
\end{pmatrix}
\begin{pmatrix}
1 & 0 & 0 & 0 \\
0 & e^{-i \pi \alpha/2^{k}} & 0 & 0 \\
0 & 0 & e^{-i \pi \beta/2^{k}} & 0 \\
0 & 0 & 0 & e^{-i \pi \gamma/2^{k-1}}
\end{pmatrix} \\
&=
\begin{pmatrix}
0 & e^{-i \pi \alpha/2^{k}} & 0 & 0 \\
e^{i \pi \alpha/2^{k}} & 0 & 0 & 0  \\
0  & 0 & 0 & e^{i \pi (\beta-2\gamma)/2^{k}} \\
0 & 0 & e^{-i \pi (\beta-2\gamma)/2^{k}} & 0
\end{pmatrix} = B.
\end{align*}
Up to logical Clifford operations, the  gate~$B$ has the following form:
\begin{align*}
\begin{pmatrix}
1 & 0 & 0 & 0 \\
0 & e^{i \pi 2\alpha/2^{k}} & 0 & 0  \\
0  & 0 &  e^{i \pi (\alpha \beta-2\gamma)/2^{k}} & 0 \\
0 & 0 & 0 & e^{-i \pi (\beta - \alpha -2\gamma)/2^{k}}
\end{pmatrix}.
\end{align*}
It is fairly straightforward to see that this has the same form as the case above when the roles of~$\alpha$ and~$\beta$ are exchanged; therefore $B \in \cal{C}_2^{(k-1)}$ by the induction hypothesis. Furthermore, since~$U_L (Z \otimes I) U_L^{\dagger} = Z \otimes I $ and~$U_L (I \otimes Z) U_L^{\dagger} = I \otimes Z $, we conclude that $U_L \in \cal{C}_2^{(k)}$, thus proving the induction hypothesis correct.
\end{proof}

It is fairly straightforward to note that the equivalent of Proposition~\ref{prop:GeneralStabilizer} will also apply in the two-qubit case. That is, the gate restrictions will also apply to general transversal operations and not just to those that are strongly transversal by using the~\textit{Decompression Lemma}.

\section{Conclusion}

Zeng \etal~classified the set of single-qubit logical transversal gates~\cite{Zeng:2007a}, showing that they must result from the application of single-qubit diagonal gates in addition to possible local Clifford operations and permutations (SWAP gates). In this work we have characterized the set of individual diagonal gates that can result in the application of a non-trivial logical gate, concluding that all of the entries must be of the form~$e^{i \pi c/2^k}$. This severely limits the set of logical gates that can be implemented in a transversal manner for qubit stabilizer codes. It also provides an important result for fault-tolerant quantum computing, as it rules out the possibility of finding transversal implementations for important gates in certain decomposition algorithms, such as the~$V$ gate. It also places restrictions on new fault-tolerance schemes which thus far have used a combination of codes to achieve fault-tolerant quantum computation.

Additionally, we have extended our analysis to two-qubit logical gates through the use of two-qubit physical diagonal gates, showing that a very similar restriction holds. In fact, in both the single and two-qubit case, the logical gates that can be implemented by transversal diagonal gate application must belong to the Clifford hierarchy, and moreover, both the single and two-qubit gates that can be implemented for a given code must reside at the same level of the hierarchy. We conjecture that this is true for all multi-qubit gates.

Open questions for future research would be to classify the set of physical diagonal gates that can implement a non-trivial logical gate for qudit systems. Additionally, it would be interesting to consider the set of logical gates that can be generated by coupling two codeblocks corresponding to different quantum error correcting codes, and determine if the same logical gate restrictions apply. Classifying the set of transversal gates for other types of codes is another interesting direction for future research which could provide insight into ways to circumvent the gate restrictions introduced in this work.

Additionally, most, if not all, magic state distillation schemes use CSS codes to distill purer magic states. These schemes use stabilizer codes with strongly transversal gates directly related to the magic state which the scheme distills. Our results suggest that magic state distillation, at least by current methods, can only distill gates in the Clifford hierarchy. Note that other gates may still be approximated by combining different magic states since we are able to distill a universal set of magic states.

\section{Acknowledgments}
T.~J.-O. would like to acknowledge the support of CryptoWorks21, FQRNT, NSERC and the Vanier Canada Graduate scholarship programs. 

J.~A. would like to acknowledge the support of ARO and Lockheed Martin Corporation. The effort depicted is supported in part by the U.~S. Army Research Office under contract W911NF-14-C-0048. The content of the information does not necessarily reflect the position or the policy of the Government, and no official endorsement should be inferred.

The authors would like to thank Bryan~Eastin, David~Poulin, and Nicolas~Delfosse for their insightful feedback on this work.

\appendix

\section{Proof of rationality}\label{app:rational}

In this case we allow each $Z(\theta_i)$ to be a $Z$ rotation about any angle, not just a rational angle. Without loss of generality we can assume each $\theta_i$ is in the range $(-1,+1)$ and $\theta_i\ne 0$ (since we can just use a new code with that qubit removed). 

Now we have a transversal gate of the form:
\begin{equation}
Z_L(\theta') := Z(\theta_1)\otimes Z(\theta_2)\otimes...\otimes Z(\theta_n).
\end{equation}

We will also assume that at least one of the angles is irrational as we have already solved the rational case. 

Constraints from $Z_L(\theta')|0_L\rangle = |0_L\rangle$ restrict as follows:
\begin{align*}
 \vec{\theta}\cdot g_{i_1}^{T} &= 0 \\
 \vec{\theta}\cdot (g_{i_1}\oplus g_{i_2})^{T} &= 0\\
 \vdots& \\
 \vec{\theta}\cdot(g_{i_1}\oplus...\oplus g_{i_{|G_X|}})^{T} &= 0,
 \end{align*}
while constraints from $Z_L(\theta')|1_L\rangle = e^{i\pi\theta}|1_L\rangle$ provide the following:
 \begin{align*}
 \vec{\theta}\cdot g_{X_L}^{T} &= \theta\\
 \vec{\theta}\cdot (g_{X_L}\oplus g_{i_1})^{T} &= \theta\\
 \vdots& \\
\vec{\theta}\cdot (g_{X_L}\oplus g_{i_1}\oplus...\oplus g_{i_{|G_X|}})^{T} &= \theta\\
& \forall 0<i_1<i_2<...<i_{|G_X|}\le|G_X|.
\end{align*}

Here the equality is taken over the Real numbers if at least one term in the sum $\vec{\theta}\cdot g_{i}^T$ is irrational, otherwise the equality is modulo some integer as before. 

Some observations:
\begin{enumerate}
  \item If $\theta_i = \frac{p}{q}\theta_j$, then if $\theta_i + \frac{p}{q}\theta_j=0 \implies \frac{\theta_i}{q}(q+p)=0 \implies \theta'(q+p)=0.$ Here $\theta'=\theta_i/q$. We can use the decompression lemma to create a new code where $Z_L$ applies $Z(\theta')$ to $p+q$ qubits.
  
   \item If $\theta_i \ne \frac{p}{q}\theta_j$, then $\theta_i + \theta_j = 0$ iff $\theta_i=0$ and $\theta_j=0$. Notice that $\theta_i$ and $\theta_j$ could be two irrational numbers which are not proportional or an irrational and a rational number (which by definition are not proportional).
   
   \item We can use these observations to reorder the qubits in the code (and possibly apply the decompression lemma to create a new code) to write $\vec{\theta}$ as $Z(1/q)\otimes... \otimes Z(1/q)\otimes Z(\theta_1)\otimes...\otimes Z(\theta_1)\otimes Z(\theta_2)...$. Here the $Z(1/q)$ are from the rational part of $Z_L$ (with $q$ a common denomonator) and $Z(\theta_i)$ are the irrational part of $Z_L$. We have used the decompression lemma to express proportional irrational angles as the same $\theta_i$. Each different $i$ corresponds to irrational angles which are not proportional. 
   
   \item Using the second observation we see that the rational angles and each set of proportional irrational angles must individually satisfy the above constraints. We have already discussed the allowable solutions given rational angles. In what follows we will show that no nontrivial solutions exist given irrational angles. 
\end{enumerate}

For each $\theta_i$ we will have constraints from $Z_L(\theta_i)|0_{L|\theta_i}\rangle = |0_{L|\theta_i}\rangle$ such that,

\begin{align*}
 \vec{\theta_i}\cdot g_{i_1|\theta_i}^{T} &= 0 \\
 \vec{\theta_i}\cdot (g_{i_1|\theta_i}\oplus g_{i_2|\theta_i})^{T} &= 0\\
 \vdots& \\
 \vec{\theta_i}\cdot(g_{i_1|\theta_i}\oplus...\oplus g_{i_{|G_X|}|\theta_i})^{T} &= 0,
 \end{align*}
while constraints from $Z_L(\theta_i)|1_{L|\theta_i}\rangle = e^{i\pi\theta}|1_{L|\theta_i}\rangle$ provide the following,
 \begin{align*}
 \vec{\theta_i}\cdot g_{X_L|\theta_i}^{T} &= \theta\\
 \vec{\theta_i}\cdot (g_{X_L|\theta_i}\oplus g_{i_1|\theta_i})^{T} &= \theta\\
 \vdots& \\
\vec{\theta_i}\cdot (g_{X_L|\theta_i}\oplus g_{i_1|\theta_i}\oplus...\oplus g_{i_{|G_X|}|\theta_i})^{T} &= \theta\\
& \forall 0<i_1<i_2<...<i_{|G_X|}\le|G_X|.
\end{align*}

Here $|0_{L|\theta_i}\rangle$ refers to the restriction to qubits which $\vec{\theta_i}$ acts nontrivially upon. Note that it is possible that $\theta=0$ for some set of proportional irrational angles. As long as $\theta\ne 0$ for some set of proportional irrational angles, then the irrational part of $\vec{\theta}$ has contributed nontrivially $Z_L(\theta')$. We will only consider the case when $\theta\ne 0$ as the other case is trivial (equivalent to applying the identity). 

Now, we will try to find a set of rows of $H_X$ and $X_L$ which satisfy all these conditions. For the underlying code to be nontrivial we require that $H_X$ has no zero columns. We assume that $a\ne 0$, otherwise the transversal operator is trivial $(Z_{L}(\theta')=I)$.

If there is only one row $h_1$ then it must be all ones and
\begin{align*}
\vec{\theta}\cdot g_{1}^{T}&=0\\
\vec{\theta}\cdot g_{X_L}^{T}&\ne 0\\
\vec{\theta}\cdot (g_1\wedge g_{X_L})^{T}&=0,
\end{align*}
but $\vec{\theta}\cdot (g_1\wedge g_{X_L})^{T}=\vec{\theta}\cdot g_{X_L}^{T}=0$ and we have a contradiction.

If $H_X$ is nontrivial and has two rows, the columns of $H_X$ are one of three types:

\begin{equation}
	a = \begin{bmatrix}
    1\\0
    \end{bmatrix},
    b = \begin{bmatrix}
    0\\
    1 
    \end{bmatrix},
    c = \begin{bmatrix}
    1\\
    1 
    \end{bmatrix}.
\end{equation}
We will refer to the combination of all columns of type $a,b, c$, by the matrix $A,B,C$, respectively.

If we have a logical operator $X_L$, then 

\begin{align*}
 \vec{\theta}\cdot(g_1\wedge g_{X_L})^{T} &= \theta(\Delta w_A + \Delta w_C) = 0,\\
 \vec{\theta}\cdot (g_2\wedge g_{X_L})^{T} &= \theta(\Delta w_B + \Delta w_C) = 0,\\
 \vec{\theta}\cdot (g_1\wedge g_2\wedge g_{X_L})^{T} &= \theta(\Delta w_C) = 0,\\
 \vec{\theta}\cdot X_L &= \theta(\Delta w_A + \Delta w_B + \Delta w_C) \ne 0.
\end{align*}

Here, $\Delta w_A = w_A^{+}-w_A^{-}$ and $w_A^{+}(w_A^{-})$ is the overlap of $A$ and $X_L$ which has support on $H_{X}^{+}(H_{X}^{-})$. Since $\theta\ne 0$ the first three constraints imply that $\Delta w_A, \Delta w_B, \Delta w_C = 0$ which imply $|X_L| = 0$ and hence a contradiction.

As we can see the proof proceeds in the same manner as in Sec.~\ref{sssec:Existence} with $w_i$ replaced by $\Delta w_i$. We reach the same contradiction given any set of proportional irrational angles and have, therefore, proven that transversal gates with single qubit rotations by irrational angles have no effect and are equivalent to applying the identity. 

\bibliographystyle{unsrt}

\end{document}